\documentclass[a4paper,10pt]{article}
\usepackage[margin=1in]{geometry}

\usepackage[authoryear,longnamesfirst]{natbib}

\usepackage{microtype}
\usepackage{graphicx}
\usepackage{subcaption}
\usepackage{tikz}
\usepackage{array,epsfig, fancyhdr,rotating}
\usepackage{algorithm}
\usepackage{algorithmic}
\usepackage{placeins}
\usepackage{multirow}
\usepackage[utf8]{inputenc} 
\usepackage[T1]{fontenc}    
\usepackage{hyperref}       
\usepackage{url}            
\usepackage{booktabs}       
\usepackage{amsfonts}       
\usepackage{nicefrac}       
\usepackage{microtype}      
\usepackage{xcolor}         
\usepackage{amsmath}
\usepackage{amssymb}
\usepackage{mathtools}
\usepackage{amsthm}
\usepackage{enumerate}
\usepackage{mathrsfs}
\usepackage{bbm}
\usepackage{enumitem}

\newcommand{\andd}{\qquad\text{and}\qquad}

\newcommand{\ind}[1]{\mathbbm{1}\left\{#1\right\}}

\newcommand{\rdd}{\mathbb{R}^{d}}
\newcommand{\re}{\mathbb{R}}

\newcommand{\Prr}[1]{\Pr\left(#1\right)}

\newcommand{\norm}[1]{|#1|}

\newcommand{\eqd}{\stackrel{\text{d}}{=}}

\newcommand{\sumn}{\sum_{i=1}^n}

\newcommand{\D}{{\rm D}}

\newcommand{\E}{\mathbb{E}}
\DeclareMathOperator*{\Proj}{ {\rm Proj}}

\DeclareMathOperator*{\argmin}{ {\rm argmin}}
\DeclarePairedDelimiter\ceil{\lceil}{\rceil}
\DeclarePairedDelimiter\floor{\lfloor}{\rfloor}
\DeclarePairedDelimiter\ip{\langle}{\rangle}

\DeclareMathOperator{\var}{Var}
\DeclareMathOperator{\EC}{EC}

\newcommand{\GS}{\mathrm{GS}}

\newcommand{\cD}{\mathcal{D}}
\newcommand{\cG}{\mathcal{G}}

\newcommand{\cA}{\mathcal{A}}
\newcommand{\cN}{\mathcal{N}}
\newcommand{\cF}{\sF}
\newcommand{\cS}{\mathcal{S}}

\newcommand{\cH}{\mathcal{H}}

\newcommand{\sH}{\mathscr{H}}
\newcommand{\sF}{\mathscr{F}}

\newcommand{\VC}{\operatorname{VC}}

\newcommand{\bfx}{\mathbf{x}}
\newcommand{\bfz}{\mathbf{z}}
\newcommand{\bfy}{\mathbf{y}}

\newcommand{\bbN}{\mathbb{N}}

\newenvironment{keywords}{%
  \par\noindent\textbf{Keywords: }\ignorespaces
}{%
  \par
}

\newtheorem{theorem}{Theorem}
\newtheorem{proposition}[theorem]{Proposition}
\newtheorem{lemma}[theorem]{Lemma}

\theoremstyle{definition}
\newtheorem{definition}[theorem]{Definition}

\newtheorem{condition}[theorem]{Condition}
\theoremstyle{remark}
\newtheorem{remark}[theorem]{Remark}

\begin{document}
\let\WriteBookmarks\relax
\def\floatpagepagefraction{1}
\def\textpagefraction{.001}



\title{Computationally tractable robust differentially private mean estimation}



%

\author{Kelly Ramsay}








\maketitle
\begin{abstract}
We develop a new, differentially private mean estimator called the balloon mean. The main features of the balloon mean are that it is computationally tractable and enjoys robustness to outlying observations. It is based on an iterative clipping procedure over expanding Mahalanobis balls, or ``balloons.'' The method satisfies zero-concentrated differential privacy and depends on a small number of interpretable tuning parameters. We provide theoretical guarantees under heavy-tailed and contaminated elliptical models, characterizing its statistical performance and robustness to outliers. Extensive simulations demonstrate that the balloon mean is robust to heavy-tailed and contaminated data, and often outperforms existing differentially private mean estimators in contaminated settings.
\end{abstract}




\section{Introduction}
Differentially private mean estimation has been an active area of research in recent years, enjoying substantial statistical and algorithmic progress, see e.g., \citep{Biswas2020,Kamath2020,Hopkins2021,Huang2021,Brown2023,2021Liub,Yu2024} and the references therein. 
On the other hand, robust mean estimation is a central problem in statistics and theoretical computer science, particularly in high dimensions and under adversarial contamination \citep{Tukey1974,Diakonikolas2016,Cheng2020,Lugosi2021}. 
Recently, several authors have begun exploring robust and private mean estimation \citep{Hopkins2021,2021Liu,2021Liub,kothari2022,Yu2024,Ramsay2025JMLR,Ramsay2025EJS,Brown2025}. 

Despite this growing literature, existing robust and private mean estimators still suffer from limitations. 
Many proposed methods are either computationally heavy, rely on sophisticated convex or sum-of-squares relaxations, or Markov chain Monte Carlo (MCMC). Others are not designed to scale well in high dimensions, or require tuning parameters that can be difficult to interpret or tune in practice. Moreover, much of the robust private literature focuses on \emph{approximate} differential privacy, whereas substantially fewer methods provide robustness guarantees under stronger privacy notions such as zero-concentrated or pure differential privacy. 
These limitations motivate the development of new estimators, and so we introduce a new differentially private mean estimator, called the \textit{balloon mean}, that is computationally efficient, depends on simple, interpretable tuning parameters, satisfies stronger forms of differential privacy, and enjoys robustness to heavy tails and some robustness to adversarial contamination.

The balloon mean can be described as follows. We start with an initial ``balloon'' or ellipse, which is defined by an initial radius, a covariance matrix, known or privately estimated, and a center. 
The data are then projected onto the balloon and a crude estimate of the mean is computed via a noisy mean of the projections, denoted by $\tilde\mu_1$. 
Then, we privately blow up an ellipse, or ``balloon'' centered at $\tilde\mu_1$ until it contains most of the data, precisely. $100\tau$\% of the data. 
We then produce a new private estimate of $\mu$, say $\tilde\mu_2$, by projecting the data onto the new balloon and computing a noisy mean of the projections. 
We iterate several times between ``balloon blowing'' and naive private mean estimation to arrive at our final estimate. 
This procedure is demonstrated in Figure~\ref{fig::algorithm-demonstration}. 

The balloon mean avoids MCMC, gradient-based optimization, and sum-of-squares subroutines, relying instead on simple linear-algebraic operations with a few interpretable tuning parameters, making it computationally appealing. It satisfies zero-concentrated differential privacy, with a natural extension to pure differential privacy if desired, providing strong privacy guarantees. Moreover, taking smaller values of the parameter $\tau$ provides robustness against adversarial contamination. We provide theoretical analysis of the balloon mean, yielding finite-sample results under heavy-tailed and adversarially contaminated elliptically distributed data in the known covariance setting. 

\begin{figure}[t]
    \centering
    \includegraphics[width=0.65\linewidth]{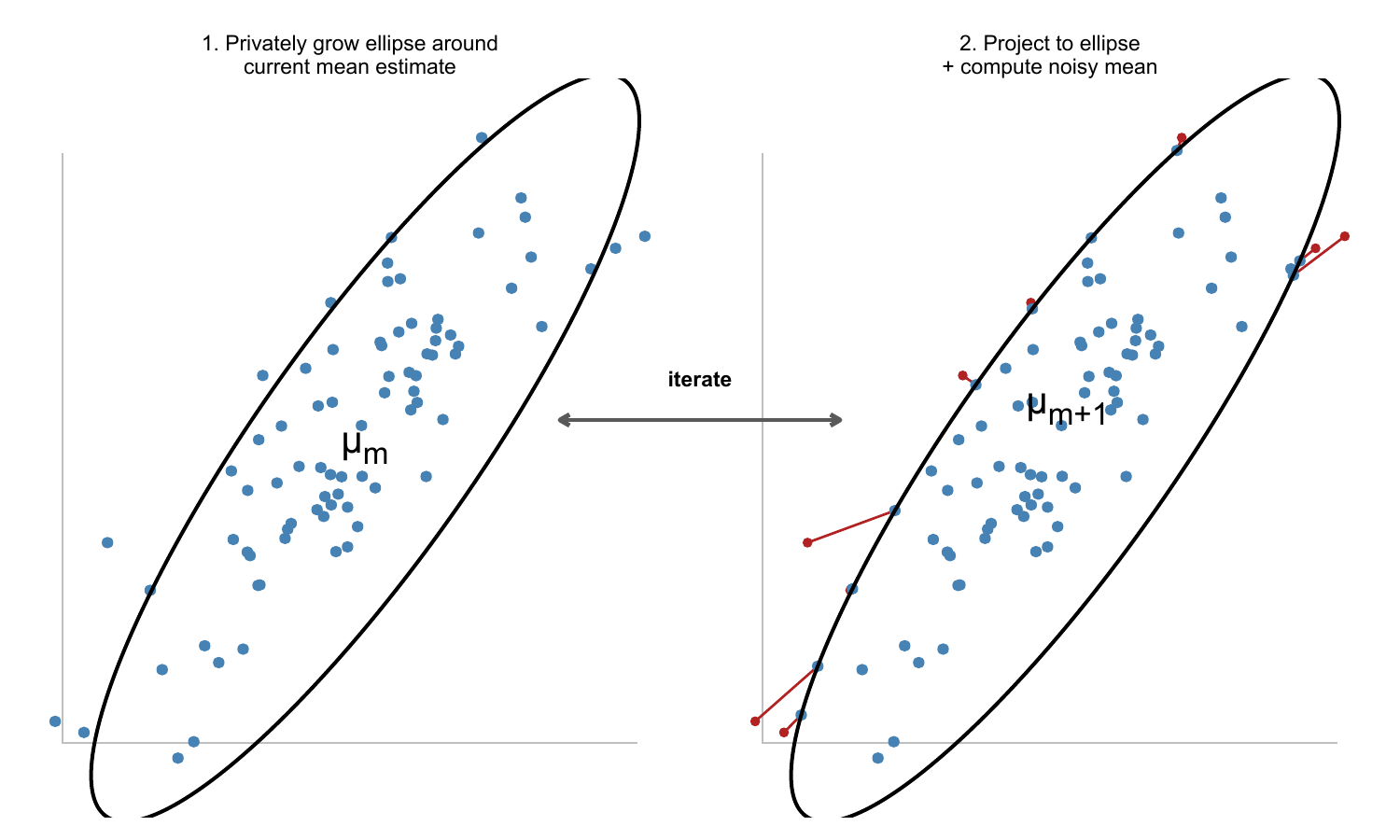}
    \includegraphics[width=0.33\linewidth]{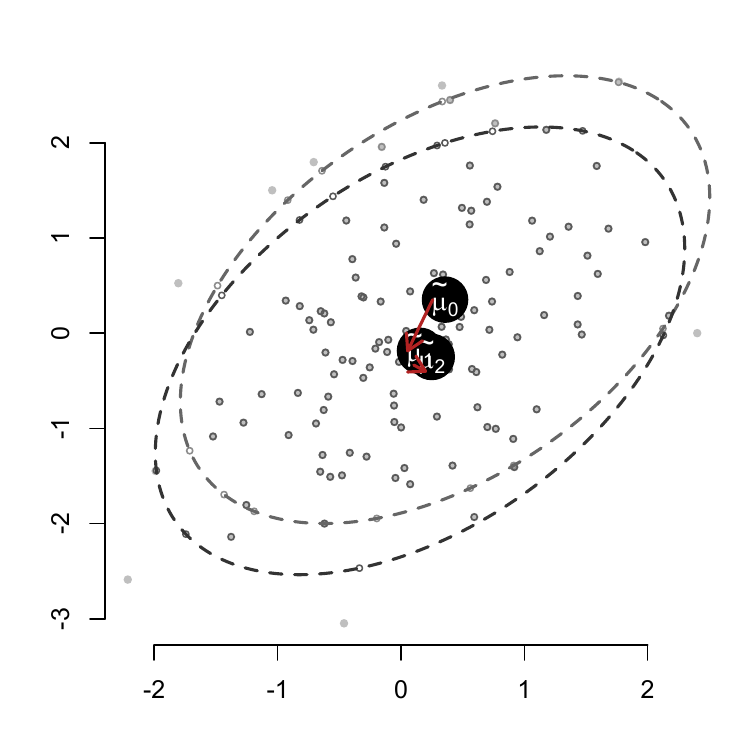}
    \caption{The balloon mean procedure illustrated. Starting from an initial balloon, or ellipse, we iterate between noisy means of projections onto the current balloon, and privately estimating a balloon that contains most of the data. The rightmost panel shows an example of two iterations of the iterated balloon mean algorithm. The dashed ellipses are the private balloons at each iteration, and successive private mean estimates are labeled.}
    \label{fig::algorithm-demonstration}
\end{figure}

\subsection{Related work}
As mentioned earlier, differentially private mean estimation is an active area of research. 
Multivariate mean estimation under approximate differential privacy has been well-studied \citep{2021Liub,2021Liu,Brown2021,Esfandiari2022,Duchi2023,Brown2023,Dagan2024,Brown2025,Ramsay2025EJS}. 
Under stronger privacy guarantees, such as pure or zero-concentrated differential privacy, many authors have considered mean estimators for a sample of subgaussian data \citep{Kamath2018, Bun2019b,Cai2019,Biswas2020}. 
Several works relax these assumptions using moment conditions, contamination models, or both \citep{Hopkins2021,Yu2024,Ramsay2025JMLR}, while maintaining stronger privacy guarantees. 
Some of these estimators rely on procedures that are computationally intensive or do not currently admit efficient implementations \citep{Hopkins2021,Ramsay2025JMLR}. 
To our knowledge, there are few computationally tractable estimators that simultaneously achieve robustness to contamination and strong privacy guarantees. 
\citet{Yu2024} provide a private, robust mean estimator based on differentially private $M$-estimation, in which robustness is achieved by minimizing a Huber loss function, through a privatized gradient descent. Rather than optimizing a robust loss, we iteratively localize the mean by privately identifying Mahalanobis balls containing most of the data and projecting observations onto these regions before averaging. 

The proposed algorithm bears some similarity to some existing algorithms, and we clarify important differences below. 
Our mean can be thought of as a multivariate, iterated private winsorized mean \citep{Ramsay2025}. 
We note that the multivariate setting is substantially more challenging, and requires several non-trivial innovations, algorithmically and theoretically, including iteration, private recentering, and handling correlation between coordinates. 
Another similar work is the \texttt{COINPRESS} mean \citep{Biswas2020}. While \texttt{COINPRESS} also computes noisy means of data clipped to centrally defined balls or ellipsoids, the radii are fixed in advance and follow a shrinking schedule, whereas our method privately and adaptively selects the radius at each iteration via a balloon blowing procedure. Lastly, the work of \citet{Huang2021} is also closely related. They introduce a differentially private mean estimator that also selects the clipping radius privately from the data. In their approach, they preprocess the data, after which the clipping radius is selected via a private quantile of the preprocessed sample norms. From there, they generate a single clipped mean, which yields an instance-optimality guarantee. In contrast, the balloon mean is inherently iterative, repeatedly recentering and projecting onto Mahalanobis balls rather than selecting a single clipping threshold. Beyond these algorithmic differences, \citet{Huang2021} focus on instance-optimality for empirical mean release, while our focus is on robustness and statistical optimality under heavy-tailed and contaminated elliptical models. 
We compare our method to both those of \citet{Biswas2020} and \citet{Huang2021} in Section~\ref{sec::sim}. 


\subsection{Contributions}

Our main contributions are: 
\begin{itemize}
    \item We introduce a novel differentially private mean estimator, the balloon mean. The balloon mean is based on iterated Mahalanobis clipping. The balloon mean is computationally simple, depends on a small number of interpretable parameters, and includes a natural robustness parameter to control resistance to outliers.
    \item We present a theoretical result, under both a heavy-tailed and an adversarially contaminated elliptical distribution model, giving explicit finite-sample error bounds in terms of dimension, sample size, privacy level, and contamination. Our result implies that in the ``heavy-tailed setting'', i.e., when the data have finite second moment, the balloon mean achieves the minimax-optimal rate of convergence up to logarithmic factors.
    \item Through extensive simulations across dimensions, privacy levels, and distributions, we show that the balloon mean is computable in high dimensions, robust to heavy-tailed and contaminated data, and outperforms existing computationally tractable private mean estimators in heavy-tailed and contaminated settings.
\end{itemize}

\section{Background}
\subsection{Problem}
We begin by formalizing the estimation problem. 
We suppose there exists a random sample of $d$-dimensional vectors, $X_1, \ldots, X_n\sim \nu$, where $\nu$ denotes the unknown population distribution or measure. 
We also suppose that we only observe an $\eta$-contaminated version of this sample, denoted $X_1', \ldots, X_n'$. That is $\floor{\eta n}$ observations in the original sample, are replaced by arbitrary values, chosen by some adversary. Importantly, these values can depend on the sample itself.
In our theoretical analysis, we assume that $\nu$ is an elliptical distribution.  
Elliptical distributions can be written in terms of three parameters: $\nu=\EC(\mu,\Sigma,F)$ , where $\mu$ is the mean, $\Sigma$ is the covariance matrix when it exists, or the scatter matrix otherwise, and $F$ is the distribution function of the length of the vector. 
In particular, if, $X\sim \EC(\mu,\Sigma,F)$, then we can write $X= \mu+\ell\Sigma^{1/2} u $ where $u$ is a uniformly random vector on the unit sphere and $\ell$ is a positive random variable which determines the length $X$, with $\ell\sim F$ and $\E\ell^2=d$. For more details on elliptical distributions, see \citep{Fang1990Symmetric}. 
This model is summarized in the following condition.
\begin{condition}\label{cond::F-inc}
Take an independent, identically distributed sample $X_1,\ldots,X_{n}\sim \nu=\EC(\mu,\Sigma,F)$, to be the uncorrupted points. Then
\begin{enumerate}[label=(\alph*)]
    \item Here, $F$ is absolutely continuous and the variance of $F$ is finite, i.e., $\int x^2dF<\infty$. 
    \item The points $X_1',\ldots,X_{n}'$ then differ arbitrarily from the uncorrupted points on at most $\eta n$ points, $0 \leq \eta \leq 1/2$, where the corrupted values may depend on the uncorrupted values. 
\end{enumerate}
\end{condition}
Note that two types of contamination are included in the model put forth in Condition~\ref{cond::F-inc}. First, assuming that only two moments exist constitutes contamination from heavy tails. A stronger form of contamination comes from the ``adversarial corruption'' model, where an $\eta$-proportion of the observations are replaced by adversarial values. Thus, we will evaluate our proposed estimator under the weaker heavy-tailed model ($\eta=0$) and the stronger adversarially corrupted model ($\eta>0$).  

\subsection{Differential privacy}\label{sec::dp}
Next, we introduce the essential concepts from differential privacy \citep{Dwork2006, Dwork2014}. 
Define a dataset of size $n\in \bbN$ to be a set of $n$ vectors in $\rdd$, and let $\cD_n$ be the set of datasets of size $n$. 
We say that a dataset $\bfx_n\in \cD_n$, is adjacent to another dataset $\bfy_n\in \cD_n$ if $\bfx_n$ and $\bfy_n$ differ by exactly one vector. 
Let $\cA_n$ be the set of pairs of adjacent datasets of size $n$. 
Let $G_{\bfx_n}$ be a probability distribution over $\rdd$ with $d\in\bbN$ that depends on $\bfx_n$. 
That is, the distribution $G_{\bfx_n}$ is a function of the dataset $\bfx_n$
If $G_{\bfx_n}$ is absolutely continuous, let $g_{\bfx_n}$ be the associated probability density function. 
For $\alpha>1$, denote the $\alpha$-R\'enyi Divergence between distributions, $F$ and $H$, by $\D_\alpha (F|H)$. 
We now define \emph{zero-concentrated differential privacy} (zCDP) \citep{dwork2016concentrated,Bun2016}. 
\begin{definition}\label{dfn::pdp}
The quantity $\theta\sim G_{\bfx_n}$ is $\rho$-zero-concentrated differentially private ($\rho$-zCDP), $\rho > 0$, if, for all $\alpha\in (1,\infty)$, 
\begin{equation*}
    \sup_{(\bfy_n,\bfz_n)\in \cA_n}\D_\alpha (G_{\bfy_n}|G_{\bfz_n})\leq \alpha\rho.
    \end{equation*}
\end{definition}
\noindent Here, $\theta$ is a $d$-dimensional differentially private quantity and $\rho$ is the \emph{privacy budget}, with smaller values of $\rho$ producing higher levels of privacy. 
The following two properties are useful for developing differentially private estimates \citep{dwork2016concentrated,Bun2016}. 
\begin{proposition}\label{prop::dp} The following hold:
\begin{itemize}
    \item \textbf{Post-processing:} For any function $f$ that does not depend on $\bfx_n$, defined on a subset of $\rdd$, if $\theta$ satisfies $\rho$-zCDP, then $f(\theta)$ satisfies $\rho$-zCDP. 
    \item \textbf{Composition:} If $\theta_1,\theta_2$ satisfy $\rho_1$-zCDP and $\rho_2$-zCDP, respectively, then $(\theta_1,\theta_2)$ satisfies $(\rho_1+\rho_2)$-zCDP.
\end{itemize}
\end{proposition}
Proposition~\ref{prop::dp} first states that functions of differentially private quantities are still private, and second that releasing two differentially private quantities results in a jointly differentially private estimate, with a larger privacy parameter. 

Next, we introduce a simple way to generate a differentially private statistic, the Gaussian mechanism \citep{dwork2016concentrated,Bun2016}. 
For a  statistic $T\colon\cD_n\to \re^p$, define the global $L_2$-sensitivity of $T$ to be 
$$\GS(T)=\sup_{(\bfy_n,\bfz_n)\in \cA_n}\norm{T(\bfy_n)-T(\bfz_n)}.$$ 
Here, $\norm{x}$ denotes the Euclidean norm of the vector $x.$
We have the following proposition.
\begin{proposition}\label{prop::adm}
Let $T\colon\cD_n\to \re^p$ be any statistic with $\GS(T)<\infty$. 
If $Z$ is a standard Gaussian random variable, then 
$$ T(\bfx_n)+\GS(T)Z/\sqrt{2\rho},$$ satisfies $\rho$-zCDP. 
\end{proposition}
\noindent Proposition~\ref{prop::adm} illustrates how correctly calibrated noise can be added to a statistic to ensure differential privacy. 

Another popular differentially private mechanism we utilize here is \texttt{AboveThreshold} \citep{Dwork2014}. In particular, we use the zero-concentrated version described by \citet{Ramsay2025}. Given a sequence of real-valued queries, \texttt{AboveThreshold} can privately identify the first query whose value exceeds a given threshold, using a privacy budget that depends only weakly on the number of queries. 
At a high level, the algorithm first ``privatizes'' the threshold by adding noise, and then processes the queries sequentially, adding fresh noise to each query answer before comparing it to the noisy threshold. Once a query is found whose noisy value exceeds the noisy threshold, the algorithm halts and outputs its index. 
In our setting, we consider a sequence of counting queries that measure how many data points fall inside a sequence of ellipses. The sequence of ellipses considered is defined by a grid parameter, to be defined in the following section.

\section{Algorithm description}
We now describe our proposed algorithm, first introducing some critical notation. 
Let $\Proj(x,R,c,A)$ be the projection of $x\in\rdd$ onto the ball of radius $R>0$ centered at $c\in\rdd$, with respect to the Mahalanobis norm $\norm{\cdot}_{A}$. The Mahalanobis norm of $x$ is $\norm{x}_{A}=\norm{ A^{-1/2}x }$. 
In addition, for a dataset $\bfx_n$, let $\Proj(\bfx_n,R,c,A)$ be the set of projections of the elements of $\bfx_n$, i.e., $\{\Proj(x_1,R,c,A),\ldots,\Proj(x_n,R,c,A) \}$. 
Let $B_{R,\Sigma}(c)$ denote the ball with respect to the norm $\norm{\cdot}_\Sigma$ of radius $R>0$ centered at $c\in\rdd$, i.e., $B_{R,\Sigma}(c)=\{\norm{x-c}_\Sigma\leq R\}$. 
We let $M$ be the number of iterations of the algorithm. 
For a given iteration $m\in\{1,\ldots,M\}$, we let $\rho_{\text{mean},m}$ and $\rho_{\mathrm{bal},m}$ denote the privacy budget allotted to the mean update and the balloon update, respectively. 
The total privacy budget is then $\rho=\rho_{\text{mean},M}+\sum_{m=1}^{M-1}(\rho_{\text{mean},m}+\rho_{\mathrm{bal},m})$. Note that the algorithm ends on the mean estimation step: no privacy budget need be assigned to the balloon update for iteration $M$.

The algorithm is initialized with a guess of the mean, $\tilde\mu_0$, and a radius $\tilde R_0$. The initial radius $\tilde R_0$ should be chosen large enough so that the true mean is within $\tilde R_0$ of $\tilde \mu_0$, i.e.,  $\mu\in B_{\tilde R_0,\Sigma}(\tilde \mu_0)$. 
Starting with the initial Mahalanobis ball, $B_{\tilde R_0,\Sigma}(\tilde \mu_0)$, we can now iterate between the mean update step and the balloon update step. We proceed first with the mean update step. The mean update step is simple, we project the observations onto the current Mahalanobis ball. The sample mean of the projections then has global sensitivity $2\tilde R_{m-1}/n$. We can therefore use the Gaussian mechanism to generate a private sample mean of these projected observations. At iteration $m$, the mean estimate is given by
$$\tilde\mu_{m}=n^{-1}\sumn \Proj(X_i',\tilde R_{m-1},\tilde\mu_{m-1},\Sigma)+Z,$$
where $$Z\sim \cN_d\!\left(0, \frac{2\tilde R_{m-1}^2}{n^2\rho_{\mathrm{mean},m}}\times\Sigma \right)\coloneqq\cN_d\!\left(0,\lambda^2_{m,n} \Sigma   \right).$$

The balloon update step proceeds as follows. The center is moved to the new mean estimate, $\tilde\mu_{m}$. We then use \texttt{AboveThreshold} to privately blow up the balloon until it contains approximately $\floor{\tau_m n}$ points, where $0<\tau_m<1$ is an input parameter that can change at each iteration. This is effectively an update to the radius. 
In more detail, let $V,V_1,V_2,\ldots$ be a sequence of independent, standard normal random variables. Denote by $\tilde\nu$ the empirical distribution of the data $X_1',\ldots,X_n'$.  
First, we supply the algorithm with: a grid size parameter $\beta>1$ and a lower bound on the radius $R_{\min}>0$. 
Given these parameters, the algorithm proceeds as follows: 
\begin{enumerate}
    \item Compute the noisy threshold $\hat \tau_m=\tau_m+V/(n\sqrt{\rho_{\mathrm{bal},m}})$, set $i=0$, and set $r_{i}=\beta^i+R_{\min} -1$. Here, $r_i$ is the current radius.
    \item Compute the noisy count of points contained in the current ellipse $B_{r_i,\tilde\mu_m,\Sigma}$. That is, compute $\hat N_i=\#\{j\colon X_j'\in B_{r_i,\tilde\mu_m,\Sigma}\} +V_i/\sqrt{\rho_{\mathrm{bal},m}}$ .
    \item If there are two few points in the current ellipse, i.e., $\hat N_i< \hat\tau_m n $, thenincrease $i$ by one and update $r_{i}=\beta^i+R_{\min} -1$. Otherwise, stop.
    \item If we stop at query $i^*$, then the updated radius is $\tilde R_m =r_{i^*}=\beta^{i^*}+R_{\min} -1$.
\end{enumerate}
The final estimate of the mean is then given by $\tilde \mu_M$. 
Pseudocode is given in Algorithm~\ref{alg:BalloonMean} and Algorithm~\ref{alg:balloonupdate}. 
A visual depiction of the algorithm, as well as an example path of the balloon mean over several iterations are given in Figure~\ref{fig::algorithm-demonstration}.

\begin{algorithm}[t]
\caption{BalloonMean}
\label{alg:BalloonMean}
\begin{algorithmic}[1]
\REQUIRE Data $\bfx_{n}$, covariance $\Sigma$, $\rho_{\mathrm{mean}},\ \rho_{\mathrm{bal}}$, grid $\beta>1$, radii $(R_{\min},\widetilde R_0)$, iterations $M$, center $\widetilde\mu_0$, target $\tau$
\ENSURE Private estimate $\widetilde\mu_M$
\FOR{$m=1,\ldots,M-1$}
\STATE $\bfy_n \gets \Proj\bigl(\bfx_n,B_{\widetilde R_{m-1},\Sigma}(\widetilde \mu_{m-1})\bigr)$
\STATE $\widetilde\mu_m \gets \left[\frac{1}{n}\sum_{y\in \bfy_n}y \;\right]+\; \cN_d\!\left(0,\lambda_{m,n}^2\Sigma\right)$
\STATE $\widetilde R_m \gets \text{BalloonUpdate}\bigl(\bfx_n,\widetilde \mu_{m},\Sigma,R_{\min},\tau_m,\rho_{\mathrm{bal},m},\beta\bigr)$
\ENDFOR
\STATE $\bfy_n \gets \Proj\bigl(\bfx_n,B_{\widetilde R_{m-1},\Sigma}(\widetilde \mu_{m-1})\bigr)$
\STATE $\widetilde\mu_M \gets \left[\frac{1}{n}\sum_{y\in \bfy_n}y \;\right]+\; \cN_d\!\left(0,\lambda_{M,n}^2\Sigma\right)$
\RETURN $\widetilde\mu_M$
\end{algorithmic}
\end{algorithm}

\begin{algorithm}[t]
\caption{BalloonUpdate}
\label{alg:balloonupdate}
\begin{algorithmic}[1]
  \REQUIRE Data $\bfx_{n}$, center $c$,\;
  covariance $\Sigma$,\;
  lower bound $R_{\min}$, grid size $\beta>1$, target $\tau$,\;
  privacy budget $\rho$
\ENSURE Estimated radius $\widetilde R$
\STATE $\hat\tau \gets \tau + \cN\!\left(0,\,(n\sqrt{\rho})^{-2}\right)$
\STATE $\hat N_i \gets 0,\;\; i \gets -1$
\WHILE{$\hat N_i/n < \hat\tau$}
\STATE $i \gets i+1$
  \STATE $N_i \gets \bigl|\{\,j : \norm{x_j-c}_{\Sigma} \le R_{\min} + \beta^i - 1\}\bigr|$
  \STATE $\hat N_i \gets N_i \;+\; \cN\!\left(0,\,\rho^{-1}\right)$
\ENDWHILE
\STATE $\widetilde R \gets R_{\min} + \beta^i - 1$
\RETURN $\widetilde R$
\end{algorithmic}
\end{algorithm}

Before proceeding to our main theoretical results, a few remarks are in order. 
First, the alternating \emph{balloon-and-update} scheme substantially mitigates the influence of poor initial values. We show in the next section that this can be achieved in few ($O(\log n)$) iterations. Furthermore, by choosing $\tau_{M-1} < 1-\eta$, the method automatically excludes up to $\floor{n\eta}$ extreme outliers, yielding some robustness to outliers, if desired. Lastly, we currently assume known covariance, but a robust, differentially private estimate of the covariance can be used in place of $\Sigma$ in practice \citep{kim2025}.

\section{Theory}
In this section, we introduce some theoretical results, mainly a non-asymptotic, high probability bound on the Mahalanobis error between the private mean estimate and the population mean.  
First, we establish privacy and computation time of the balloon mean estimator. 
\begin{proposition}\label{prop::priv-compu}
The balloon mean satisfies $\rho$-zero concentrated differential privacy and can be computed in $O(d^3+nd^2+Mnd+Mn\log n)$ time. 
\end{proposition}
The proof of Proposition~\ref{prop::priv-compu} can be found in Appendix~\ref{app::proofs}. 
In order for our statistical results to hold, we later require that $n>d$ and we recommend choosing $M\propto\log n$. Under these conditions, the computational cost reported in Proposition~\ref{prop::priv-compu} is dominated by $O(nd^2)$. To achieve this computational cost, the \texttt{BalloonMean} and \texttt{BalloonUpdate} can be implemented more efficiently than the pseudocode given in Algorithms~\ref{alg:BalloonMean} and \ref{alg:balloonupdate}. The whitening and matrix decomposition need only be performed once, and \texttt{BalloonUpdate} can be implemented similarly to the quantile procedure of \citet{Durfee2024}. The presented pseudocode is for ease of exposition. 

Furthermore, concerning the case of very high dimensional estimation, we note that the $O(d^3)$ term in Proposition~\ref{prop::priv-compu} arises only from the one-time inversion and square-root decomposition of the covariance matrix. This cost is not incurred at every iteration. In settings where the covariance is identity, diagonal, sparse, or otherwise structured, the whitening step can be carried out more efficiently, and the dominant cost becomes the nearly linear-in-$n$ iterative part of the algorithm. Conversely, for arbitrary dense covariance matrices, any method that aims to control Mahalanobis error uniformly must either perform whitening or use an equivalent adjustment. Thus, the $O(d^3)$ term should be viewed as the cost of working in a general dense covariance setting rather than as an intrinsic feature of the balloon mean.

We now introduce three conditions, which, in addition to Condition~\ref{cond::F-inc}, are necessary for our main statistical result. 
\begin{condition}\label{cond::mean-ball}
   It holds that $|\mu_0-\mu|_{\Sigma}\leq \gamma \tilde R_0$ for some universal constant $0<\gamma<1-1/\sqrt{3}$. 
\end{condition}
Condition~\ref{cond::mean-ball} requires that our initial balloon captures $\mu$, and that $\mu$ is not arbitrarily close to the boundary of the initial ball $\tilde R_0$. 
Our main result yields that we can take this initial balloon's radius to be very large, making the Condition~\ref{cond::mean-ball} non-restrictive. 
For instance, we take $\tilde R_0\gtrsim \sqrt{n}$.  

Let $\Delta_R=\tilde R_0 -R_{\min}+ 1$ and $\delta \in (0,1)$ be a user-specified failure probability parameter. 
We consider the following choice of $\tau$,
\begin{equation}\label{eqn::tau2}
   \tau^*=1-\left(8\eta+988\frac{\log\left(36\left(\frac{\log \Delta_R}{\log\beta}+ 1\right)\kappa^*\log n/\delta\bigvee e\right)}{n}+16384\frac{d+1}{n} \right),
\end{equation}
and we let $\tau'\coloneqq\tau'(\delta)=1-\tau^*(\delta)$. 
That is, $\tau'$ is the approximate fraction of data that lies outside the balloon. 
In \eqref{eqn::tau2} $\kappa^*=w/\log(10/9)+1/\log n >0$ is a constant, which depends on the universal constant $w>0$ given in Condition~\ref{cond::beta}. 

The next condition places restrictions on the input parameters. For $q\in[0,1]$, $c\in\rdd$, let $\xi_{q,c}=\inf\{R\colon \nu(B_{R,\Sigma}(c)^c)\leq q\}$. In words, $\xi_{q,c}$ is the smallest radius for which a ball centered at $c$ captures at least a $1-q$ fraction of the probability mass of $\nu$. We can refer to $\xi_{q,c}$ as the \textit{minimal radius} at level $q$, centered at $c$. Note that the radius is decreasing in $q$, for fixed $c$.
\begin{condition}\label{cond::beta}
It holds
\begin{enumerate}[label=(\alph*)]
    \item {\small$1<\beta\leq 1+\inf_{|c-\mu|_{\Sigma}\leq 2\tilde R_0}\frac{\xi_{3\tau'/4,c}-\xi_{5\tau'/4,c}}{\xi_{5\tau'/4,c}-R_{\min}+1}\coloneqq 1+b_n,$}
    \item $R_{\min}\leq \inf_{c\in\rdd}\xi_{2\tau',c}/2$ and $  2\sqrt{n}\leq \tilde R_0\leq n^{w}$ for some universal constant $w>1/2$. 
\end{enumerate}
\end{condition}
Condition~\ref{cond::beta} part (a) says that the grid must be fine enough in order for the estimated mean to be close to $\mu$. There must be some distance between the minimal radius at level $3\tau'/4$ and at $5\tau'/4$. This will hold if $F$ is absolutely continuous with connected support. 
The denominator is positive by Condition~\ref{cond::beta} part (b). Condition~\ref{cond::beta} part (b) gives bounds on the minimum and maximum radius input parameters. Notice that we can input a maximum radius increasing in $n$, which implies we do not need an initial ball that is small or centered close to $\mu$. 

Our last condition is a small technicality that concerns ties. 
\begin{condition}\label{cond::distinct}
    At any interation of the algorithm $m\in 1,\ldots,M-1$, the distance between the points $X_1',\ldots,X_n'$ and $\tilde \mu_{1},\ldots, \tilde \mu_{M-1}$ are not exactly equal to any of the grid values $\beta^i+R_{\min}-1$.
\end{condition}
Condition~\ref{cond::distinct} is used to avoid the adversary stacking $\eta n$ points on top of a grid point, which can cause issues in the \texttt{AboveThreshold} algorithm. To achieve this in implementations, one can add a small, continuous random jitter to the chosen $\beta$ value, or equivalently to the queried radii. 
The jitter could be added to the algorithm, and the condition removed, but it complicates the exposition. 

We now present our main result. 
Let $\rho_-=\min_{m\in[M]}\rho_{\mathrm{mean},m}$ and note that we say write $a\lesssim b$ ( $a\gtrsim b$ ) if $a\leq b$ ($a\geq Cb$) for a universal constant $C>0$. 
\begin{theorem}\label{thm::main-result}
Suppose that Condition~\ref{cond::F-inc}, and Conditions~\ref{cond::mean-ball}--\ref{cond::distinct} hold, $\tau_1=\ldots=\tau_M=\tau^*$, $M=\ceil{\kappa^* \log n}$, and $\inf_{m\in[M-1]}\rho_{\mathrm{bal},m} \geq 16/24349$. Then, given $0<c_0  e^{-d/c_1}<\delta/M<1$ and $\rho_->0$ such that, 
\begin{equation}\label{eqn::ncond}
    n\geq c_2(w)\frac{d\log n}{\rho_-\wedge 1}\vee c_3\left(\log\left(c_5w\left(\frac{\log\Delta_R}{\log\beta}+1\right)\log n /\delta\right)\right),
\end{equation}
and $\eta\leq 1/c_4$ for sufficiently large constants $c_0,c_1,c_2(w),c_3,c_4,c_5>0$, then there exists a constant $c_6>0$ such that with probability at least $1-\delta$,
\begin{equation}\label{eqn::main-result}
    |\tilde \mu_M-\mu|_{\Sigma}\lesssim\sqrt{d\eta}+\sqrt{\frac{d\log n}{n(\rho_-\wedge 1)}}
   +\sqrt{\frac{ \log\left(c_6 w\left[\left(\frac{\log\Delta_R}{\log\beta}+ 1\right)\vee \log n\right]/\delta\right)}{n}}    .
\end{equation}
\end{theorem}
The proof of Theorem~\ref{thm::main-result} can be found in Appendix~\ref{app::proofs}. 
Theorem~\ref{thm::main-result} establishes an upper bound on the rate of convergence of the balloon mean.  
A lower bound on the achievable rate can be obtained by combining some non-private and private lower bounds, \citep{Kamath2020,Lugosi2021}, yielding
\begin{equation}\label{eqn::lb}
      |\tilde \mu-\mu|_{\Sigma}\gtrsim\sqrt{\frac{d}{n(\rho_-\wedge 1)}}+\sqrt{\frac{\log(1/\delta)}{n}} +\sqrt{\eta},
\end{equation}
for any private estimator $\tilde \mu$ of $\mu$ and $\nu$ as in Condition~\ref{cond::F-inc}. 
Under the weaker heavy-tailed contamination model, i.e., $\eta=0$, if the tuning parameters are chosen so that the term $\left({\log\Delta_R}/{\log\beta}+ 1\right),$ is at most polynomial in $n$ and $d$, the rate \eqref{eqn::main-result} is minimax optimal up to logarithmic factors \citep{Kamath2020}. 
When $\eta>0$, the contamination contribution scales as $\sqrt{d\eta}$. 
Based on existing lower bounds for adversarial contamination \citep{2021Liub,Lugosi2021}, we expect the lower bound given in \eqref{eqn::lb} to be tight. This suggests that our current contamination dependence in \eqref{eqn::main-result} is suboptimal by a factor of $\sqrt{d}$ for the stronger adversarial contamination model.

We believe this gap is structural. To see this, based on its construction, the balloon mean does not explicitly account for the case where the contaminating points are placed along a single direction, inside the thin shell that contains the bulk of the data in high dimensions, see \citep{Diakonikolas2016}. 
It is well known that simultaneously achieving optimal adversarial contamination dependence, strong privacy guarantees, and computational efficiency is highly challenging, and to our knowledge, no existing algorithm fully satisfies all three. 
Our focus is on developing a method that is simple to implement while still enjoying optimal robustness under the weaker, heavy-tailed model, and characterizing its rate under the stronger, contaminated model. 
Improving the contamination dependence while preserving privacy and computational tractability is an important direction for future work.

Note that the dependence of the upper bound \eqref{eqn::main-result} on the rate of convergence on the algorithmic input parameters $\beta$, $\tilde R_0$ and $R_{\min}$ is logarithmic. This weak dependence indicates that the method is not sensitive to fine-tuning and can be used in practice without extensive or even any tuning of these parameters. 
We corroborate this behavior empirically in Section~\ref{sec::sim}.


\begin{remark}[Choice of $\tau$]
The constants appearing in \eqref{eqn::tau2} are not optimized.  
They arise from worst-case concentration bounds and union arguments used in the analysis, and are therefore very conservative. 
The theory provides guidance on the role and scaling of $\tau_m$, but it is not necessary to take $\tau_1,\ldots,\tau_M$ exactly as in \eqref{eqn::tau2}. 
In practice, $\tau_m$ should be viewed as a robustness tuning parameter that controls the tolerated fraction of extreme observations. It is analogous to the trimming proportion in a trimmed mean. Thus, if approximately 30\% contamination is expected, one may choose $\tau\approx$0.7, with some buffer for the private noise, whereas in clean settings $\tau$ can be chosen close to one (e.g., $\tau$=0.9 as a practical default for moderate sample sizes, according to the empirical study, see Section~\ref{sec::sim}). 
\end{remark}
\section{Empirical Evaluation}\label{sec::sim}
\subsection{Sensitivity to tuning parameters}\label{sec::tuning}
All code from all experiments can be retrieved from the anonymized \href{https://anonymous.4open.science/r/BalloonMean-8822/README.md}{Github Repository}. 
The code supports GPU computation but it is not required to run the experiments. 
We conducted an extensive simulation study examining sensitivity to all tuning parameters across dimensions, sample sizes, and privacy levels. We varied one parameter at a time while fixing the others at default values. We considered two threshold settings in all simulations: a high-$\tau$ variant meant for situations where robustness is not needed and a robust low-$\tau$ variant. 
The default values were given by $\tilde \mu_0=(0,\ldots,0)^\top$, $R_{\min}=1$, $\tilde R_0=50\sqrt d$, $\beta=1.01$, $M=4$, $\tau_{1}=\tau_2=\tau_3=0.9$ for the high-$\tau$ variant, and $\tau_1=0.9$, $\tau_2=0.8$ and $\tau_3=0.6$ for the low-$\tau$ variant. We assessed the effect of $\tilde \mu_0,\tilde R_0,\beta,M$ and $\tau$. Additionally, we looked at the effect of changing $\tau$ at each iteration, i.e., the $\tau$-schedule. 
We did not look at $R_{\min}$, as the analogous parameter in similar works showed little influence on the algorithm \citep{Durfee2024,Ramsay2024,Ramsay2025}. 
We considered four population distributions: (i) a Gaussian model, (ii) a contaminated Gaussian model in which an $\eta=0.1$ fraction of the data was replaced by adversarial outliers obtained by shifting the mean and inflating the covariance, (iii) a multivariate $t$ distribution with three degrees of freedom, and (iv) a non-Gaussian ``banana-shaped'' distribution. In each trial, the true mean was $\mu=(1,\ldots,1)^\top$. 
Here we summarize the main trends while full results and a complete description of simulation parameters are deferred to Appendix~\ref*{app::tuning} and Appendix~\ref*{app::tau-schedule} in the supplementary material.  
Performance was measured using Mahalanobis error with respect to the true mean.

Overall, we found that the balloon mean exhibits low sensitivity to its main tuning parameters. In particular, performance is minimally affected by varying $\tilde \mu_0$, $\tilde R_0$ and the grid size $\beta$. For the number of iterations, in low dimensions, the error appears to stabilize at three iterations ($M=3$), whereas in high dimensions, the error stabilizes at four iterations ($M=4$). This is consistent across generating distribution, sample size $n$, and does not depend on the $\tau$ variant. 
We found that in high-dimensions smaller $\tau$ values $\tau<0.9$ generally perform better. This is consistent across generating distribution and $\tau$ variant. For all dimensions, in the contaminated settings (heavy-tailed and contaminated Gaussian) smaller $\tau$ values perform better, as expected. In clean settings with low and moderate dimensions, the dependence on $\tau$ is weaker and sometimes reverses in very low dimensions. 

As for the $\tau$-schedule, for the high-$\tau$ (non-robust) variant, decreasing $\tau$ across iterations consistently improves performance in contaminated settings, reflecting the increased robustness induced by smaller $\tau$. In Gaussian and banana-shaped settings, decreasing schedules are also preferred in higher dimensions and under stronger privacy, while in other settings the three schedules perform similarly. 
For the low-$\tau$ (robust) variant, decreasing schedules provide only modest additional gains in contaminated settings, and little difference is observed in the clean data models.

\subsection{Comparison to existing estimators}
\begin{figure}[t]
    \centering
    \includegraphics[width=\textwidth]{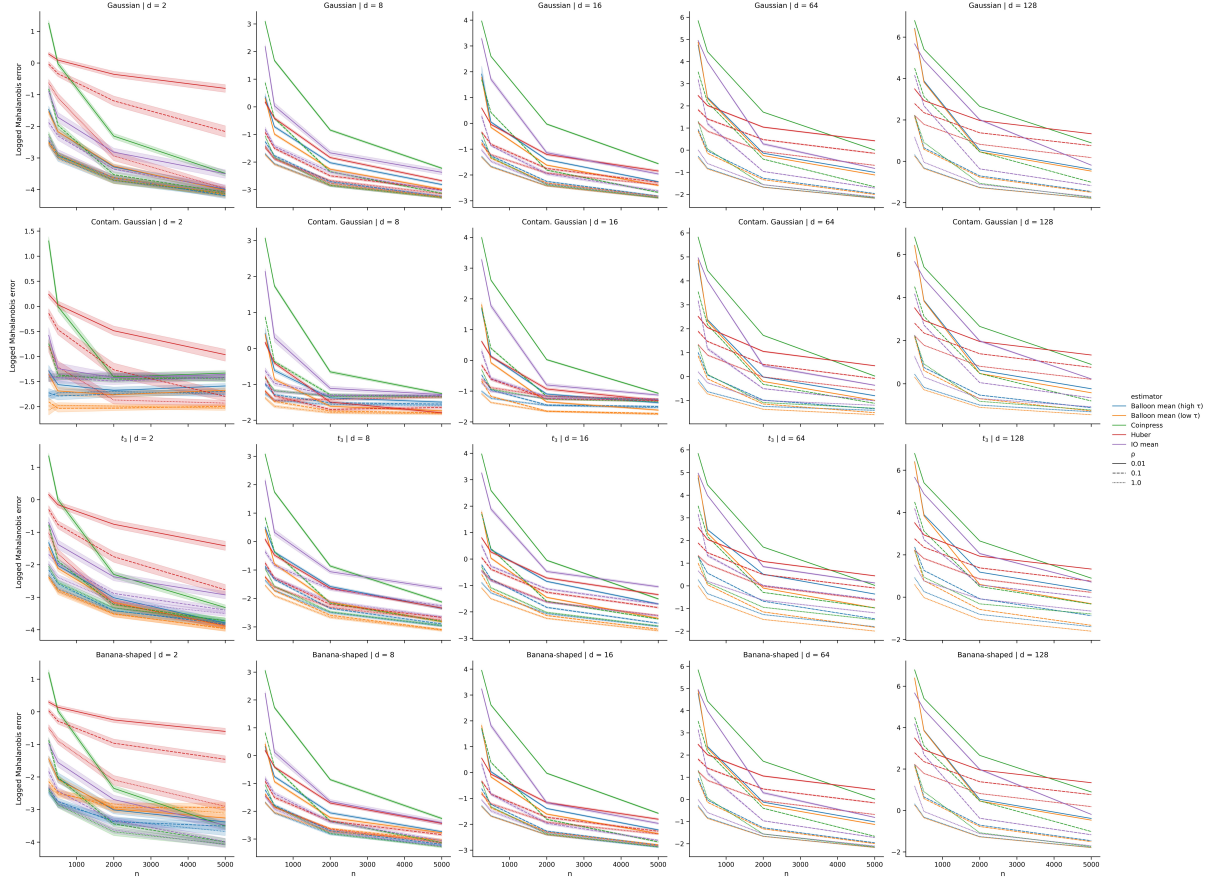}
    \caption{Logged Mahalanobis error as a function of the sample size $n$ for mean estimation under four distributions (Gaussian, contaminated Gaussian, $t_3$, and banana-shaped) and dimensions $d \in \{2,8,16,64,128\}$. Solid lines show the mean error over repetitions, with shaded bands indicating a 95\% confidence interval over the simulation runs. We compare the balloon mean with high-$\tau$ and low-$\tau$ settings against several private mean estimators, across privacy levels $\rho \in \{0.01,0.1,1.0\}$ (linetypes). The balloon mean is stable across dimensions and distributions, with the low-$\tau$ variant providing improved robustness under contamination and heavy tails, while remaining competitive in uncontaminated settings.}
    \label{fig:full-results}
\end{figure}

We now compare the balloon mean to some existing estimators satisfying zCDP, and therefore do not include approximate differentially private robust methods, which operate under a weaker different privacy regime.
For each dataset, we computed the non-private sample mean as a baseline, together with five private mean estimators: (i) low-$\tau$ variant of the balloon mean, (ii) high-$\tau$ variant of the balloon mean, (iii) \texttt{COINPRESS} \citep{Biswas2020}, (iv) the private Huber $M$-estimator \citep{Yu2024}, and (v) the instance optimal mean \citep{Huang2021}.
We considered a range of dimensions $d\in\{2,8,16,64,128\}$, sample sizes $n\in\{250,500,1000,2000,5000\}$, and privacy budgets $\rho\in\{0.01,0.1,1\}$. Dimensions in powers of 2 were used to accommodate the instance optimal mean. 
We tested the same population distributions as in Section~\ref{sec::tuning}.
For each combination of parameters, we conducted 200 trials. 

For methods that required it, the initial center was taken to be the origin, and the starting radius was $50\sqrt{d}$. 
We set a grid size of $\beta=1.01$ and we chose $M=4$ with two threshold settings: high-$\tau$ variant $(\tau_{1},\tau_2,\tau_3)=(0.9,0.9,0.9)$ and a low-$\tau$ variant $(\tau_{1},\tau_2,\tau_3)=(0.7,0.6,0.5)$. 
For \texttt{COINPRESS}, we used privacy-splitting weights proportional to $(1/3,\,1/2,\,1)$ for the three iterations, as recommended in the original work. 
The private adaptive Huber estimator was run using $T=\lfloor\log n\rfloor$ iterations \citep{Yu2024}. 
The default parameters $T=10$ and \texttt{prop}=0.25 were taken for the instance optimal mean \citep{Huang2021}. 
\texttt{COINPRESS} and the balloon mean were supplied with the known covariance, the Huber and instance-optimal baselines were run using the implementations recommended by their authors.
All tuning parameters were held fixed across experiments and were not optimized for individual distributions or methods. 
The balloon mean hyperparameters for the balloon mean were selected a priori using simple heuristics, the interpretation of the parameters, and loose guidance from the theoretical analysis, they were not chosen based on the sensitivity study.

Figure~\ref{fig:full-results} presents  the full set of results. 
We see that smaller values of $\tau$ lead to a more robust balloon mean: the low-$\tau$ variant performs better under contamination and heavy-tailed distributions, while the two balloon means behave similarly in the clean Gaussian setting. The Huber estimator exhibits some robustness to heavy tails and contamination in low dimensions, but its performance degrades in higher dimensions. The \texttt{COINPRESS} mean and the instance optimal mean remain reasonably stable under contamination and heavy tails, though they are less robust than the low-$\tau$ variant balloon mean. Overall, the balloon mean performs well relative to the competing methods, particularly in higher-dimensional settings and under stronger privacy constraints, making it competitive with, and even in several settings improving upon, state-of-the-art approaches.

\section{Summary of results and future directions}
The balloon mean occupies a practical middle ground in robust differentially private mean estimation. The empirical results support the performance of the method, and demonstrate it is competitive and sometimes outperforms existing methods. The theoretical results give it statistical credibility, especially in the weaker heavy-tailed contamination model. Furthermore, the balloon mean is not highly sensitive to several of its tuning parameters, while the other tuning parameters have practical interpretations. These findings suggest that the method can be used with a small set of default choices, rather than requiring extensive tuning for each new problem. 

Several limitations should be emphasized. First, the main theoretical results are developed in the known covariance setting. This assumption is useful for isolating the behavior of the mean estimator and demonstrating the merit of the iterative Mahalanobis clipping idea, but in many applications the covariance must also be estimated privately and robustly. Second, the theory relies on elliptical structure, which provides a tractable framework for analyzing Mahalanobis clipping but does not cover all heavy-tailed distributions of practical interest. 
Finally, the balloon mean may perform suboptimally if a cluster of outliers is adversarially placed along one direction in high dimensions. 

Future work should therefore proceed in several directions. A first direction is to generalize the results to the setting of unknown covariance. 
A second direction is to improve its performance if a cluster of outliers is adversarially placed along one direction in high dimensions by incorporating a private filtering step. 
The authors are also currently exploring how modern mean estimation can be used to improve private learning in the private setting. 

Overall, the balloon mean provides a simple and effective approach to robust private mean estimation. Its combination of privacy, robustness, computational tractability, and interpretable tuning makes it a promising tool for high-dimensional private estimation, while also opening several directions for further theoretical and empirical development.

\appendix

\section{Technical Proofs}\label{app::proofs}
For this section, we slightly redefine $\tau$
\begin{equation}\label{eqn::tau}
    \tau_1,\ldots,\tau_M=\tau=1-\left(8\eta+988\frac{\log\left(3\left(\frac{\log \Delta_R}{\log\beta}+ 1\right)M^*/\alpha \bigvee e\right)}{n}+16384\frac{d+1}{n}\right),
\end{equation}
and we let $\tau'\coloneqq\tau'(\alpha,M^*)=1-\tau(\alpha,M^*)$. In Theorem~\ref{thm::main-result} we take $M^*=\ceil{\kappa^*\log n} $ and $\alpha=\delta/12$ as a special case of \eqref{eqn::tau}, matching $\tau^*$ in \eqref{eqn::tau2}.
In addition, we let $h(\beta)=\log \Delta_R/\log\beta,$
for brevity. 

\paragraph{Overview of proof strategy:} The proof is long and contains many lemmas, so it is helpful to give an overview of the strategy and the related lemmas that we need. The proof proceeds by showing that the crude estimate $\tilde\mu_1$ is not ``too bad'' with high probability, but has linear dependence on the input radius (Lemma~\ref{lemma::init_mean}). We then need to prove that over iterations, this dependence is successively reduced. To do that, we show that in any one iteration $m$, we make some small error, say $a_n$, plus some portion of the previous iteration's error, with high probability (see Lemma~\ref{lemma:one-iteration}). In other words, we prove that for some $0<p_n<1$, with high probability,
$\norm{\tilde\mu_m-\mu}_\Sigma\leq a_n+p_n\norm{\tilde\mu_{m-1}-\mu}_\Sigma$ (Lemma~\ref{lemma:one-iteration}). Summing this out over iterations yields a geometric series, yielding $p_n^{M-1}\tilde R_0$, which allows us to resolve this with $\log \tilde R_0$ iterations. 
Most of the effort is concentrated in Lemma~\ref{lemma:one-iteration}, which requires a number of helper lemmas, the two most important being a high probability bound on the random radius $\tilde R_{m}$ (Lemma~\ref{lemma::radius_error}) and an empirical-process bound for the clipped means (Lemma~\ref{lemma::dyadic-shell}).


\begin{proof}[Proof of Theorem~\ref{thm::main-result}]
Assume without loss of generality $\mu=0$ and take $M^*=M$. 
We first bound $|\tilde \mu_1|_\Sigma$. 
Observe that the conditions of Lemma~\ref{lemma::init_mean} are satisfied. 
Let $S(x)=1-F(x)$.  
Hence, applying Lemma~\ref{lemma::init_mean} in combination with Markov's inequality and the fact that $\tilde R_0>2\sqrt{n}\gtrsim \sqrt{n (1-\gamma)^2}$ for $\gamma<1-1/\sqrt{3}$ from Condition~\ref{cond::mean-ball} and Condition~\ref{cond::beta} part (b) yields that with probability at least $1-3\alpha/M^*$, we have that,
\begin{equation}\label{eqn::init-mean}    
  |\tilde\mu_{1}|_{\Sigma}\lesssim   \tilde R_0\left(\frac{\sqrt{d}}{n(\sqrt{\rho_{-}}\wedge 1)} 
    + \sqrt{\frac{\log(2M^*/\alpha)}{n}}+\eta\right),
\end{equation}

We now bound $\norm{\tilde \mu_m}_\Sigma$ for iterations $m\geq 2$. 
We want to apply Lemma~\ref{lemma:one-iteration}, and so we must show that $|\tilde \mu_{1}-\mu|_\Sigma \leq 3\tilde R_0/4$. 

For this, we proceed via induction. 
Now, noting that $\tilde R_0>1$, we have that for large enough constants $C,C'>0$, if $n\geq C\left( \sqrt{d /(\rho_{-}\wedge 1)}\vee \log(M^*/\alpha)   \right)$ and $\eta<C'$, we have that \eqref{eqn::init-mean} yields that $ |\tilde \mu_1|_\Sigma \leq 3\tilde R_0/4$ with probability at least $1-3\alpha/M^*$. 
This condition on $n$ holds by assumption. 
We now show that by induction, $|\tilde \mu_{m}-\mu|_\Sigma \leq 3\tilde R_0/4$. We have already proven the base case, and we move to the induction step. 

We now assume that $|\tilde \mu_{k-1}-\mu|_\Sigma \leq 3\tilde R_0/4$ for some $2\leq k\leq M$. 
Using the induction assumption, we can apply Lemma~\ref{lemma:one-iteration}. 
Doing so, and letting $$A=d\log n\vee  \log(M^*c_6[(h(\beta)+1)\vee \log n]/\alpha),$$ 
for some large universal constant $c_6>0$, yields that with probability at least $1-9\alpha/(2M^*)$ we have
\begin{equation*}
    |\tilde \mu_k|_{\Sigma}\lesssim \left(1+\norm{\tilde \mu_{k-1}}_\Sigma\right)\left( \sqrt{\frac{d}{n\rho_{\mathrm{mean},k}}}+\sqrt{\frac{ A}{n}}     +\sqrt{\eta}\right)+\sqrt{d\eta}.   
\end{equation*}
Recalling that $\rho_-=\min_{k\in[M]}\rho_{\mathrm{mean},k}$, we want to then show that 
\begin{equation*}
    \left(1+\norm{\tilde \mu_{k-1}}_\Sigma\right)\left( \sqrt{\frac{d}{n\rho_-}}+\sqrt{\frac{ A}{n}}     +\sqrt{\eta}\right)+\sqrt{d\eta}\leq 3\tilde R_0/4.
\end{equation*}
Using the conditions on $\eta$ and $d$ and the the induction assumption, this is implied by the inequality
\begin{equation*}
\sqrt{\frac{d}{n\rho_-}}+\sqrt{\frac{ A}{n}}     +\sqrt{\eta}\leq K_1,
\end{equation*}
for a sufficiently small constant $K_1>0$. 
Now, this inequality is directly implied by the conditions on $n$ and $\eta$ given in the theorem, provided that the constants are chosen large enough. 
Therefore, by induction and a union bound, we have that with probability at least $ 1-9\alpha/2$, $ |\tilde\mu_{m}|_{\Sigma}\leq 3\tilde R_0/4$ simultaneously, for all $m\in[M]$. 
Note that we have also shown on the same event, there exists a universal constant $D_1>0$ for all $2\le m\le M$, we have
\begin{equation*}
          |\tilde \mu_m|_{\Sigma}\leq D_1\Bigg(\left(1+\norm{\tilde \mu_{m-1}-\mu}_\Sigma\right)\left(\sqrt{\frac{d}{n\rho_{-}}}+ \sqrt{\frac{A}{n}}
    +\sqrt{\eta}\right)+\sqrt{d\eta}\Bigg). 
\end{equation*}


Next, taking the constants $c_1$ through $c_4$ large enough, we have that
$$\sqrt{\frac{d}{n\rho_{-}}}+\sqrt{\frac{A}{n}}
    +\sqrt{\eta}\leq \frac{9}{10 D_1}.$$  
The previous two inequalities combined, along with the properties of the geometric series, yield that there exists a universal constant $D_1>0$ such that with probability at least $1-9\alpha/2 $
\begin{equation}\label{eqn::unrolled}
    |\tilde \mu_M|_{\Sigma}\leq  D_1\left(\sqrt{\frac{d}{n\rho_-}}
   +\sqrt{\frac{A}{n}}
    +\sqrt{d\eta}\right)+|\tilde \mu_1|_\Sigma \left(\frac{9}{10}\right)^{M-1}. 
\end{equation}

Now, combining the previous bound with \eqref{eqn::init-mean}, and taking $M>\log(\tilde R_0)/\log(10/9)+1$ yields $\left(9/10\right)^{M-1}<1/\tilde R_0$. 
Then, we have that with probability at least $1-9\alpha/2$, 
\begin{equation*}
      |\tilde \mu_M|_{\Sigma}\lesssim
   \sqrt{\frac{d}{n\rho_-}}
   +\sqrt{\frac{A}{n}}
    +\sqrt{d\eta}
  +  \frac{d}{(\tilde R_0-|\mu_0|_{\Sigma})^2}
    +\frac{\sqrt{d}}{n\sqrt{\rho_{-}}} +\eta
    + \sqrt{\frac{\log(2M^*/\alpha)}{n}}.
\end{equation*}
Simplifying, and using Condition~\ref{cond::mean-ball} and Condition~\ref{cond::beta}, we have that with probability at least $1-9\alpha/2$,
\begin{equation*}
     |\tilde \mu_M|_{\Sigma}\lesssim \sqrt{\frac{d\log n}{n(\rho_-\wedge 1)}}
   +\sqrt{\frac{  \log(M^*c_6[(h(\beta)+1)\vee \log n]/\alpha)}{n}}
    +\sqrt{d\eta}. 
\end{equation*}
Taking $M^*=\ceil{\kappa^*\log n}$ where  $\kappa^*\log n>\log(\tilde R_0)/\log(10/9)+1$ holds from the definition of $\kappa^*$ and Condition~\ref{cond::beta}, along with substituting in $\alpha=\delta/12$ yields the final result. 
\end{proof}
\subsection{Contraction and error of a single iteration}
\begin{lemma}\label{lemma:one-iteration} 
Suppose that Condition~\ref{cond::F-inc}, and Condition~\ref{cond::mean-ball}--\ref{cond::distinct} hold, $2\leq m\leq M$, $\rho_{\mathrm{bal},m-1} \geq 16/24349$, $\tau=\tau(\alpha,M^*)$ as in \eqref{eqn::tau}, $|\tilde \mu_{m-1}-\mu|_\Sigma \leq 3\tilde R_0/4$, 
given $c_0  e^{-d/c_1}<\alpha/M^*<2/9$, 
$$n\geq c_2(w)(d\log n\vee 1)\vee c_3\log(3(h(\beta)+1)M^*/\alpha) \andd \eta\leq 1/c_4,$$ 
for sufficiently large constants $c_0,c_1,c_2(w),c_3,c_4>0$, then with probability at least $1-9\alpha/2M^*$ we have that there exists $c_5>0$ such that
\begin{multline*}
      |\tilde \mu_m-\mu|_{\Sigma}\lesssim \\ \left(1+\norm{\tilde \mu_{m-1}-\mu}_\Sigma\right)\left(\sqrt{\frac{d}{n\rho_{\mathrm{mean},m}}}+ \sqrt{\frac{d\log n\vee  \log(M^*c_5[(h(\beta)+1)\vee \log n]/\alpha)}{n}}
    +\sqrt{\eta}\right) +\sqrt{d\eta}.     
\end{multline*}
\end{lemma}
\begin{proof}
Without loss of generality, assume that $\mu=0$, since an iteration of the algorithm is translation invariant (with probability 1). 
Take the constants in the conditions on $n$ and $\eta$ to be large enough such that $0<\tau'<1/2$. 
We have that for $X\sim\EC(0,\Sigma,F)$ independent of the sample, 
\begin{align*}
|\tilde \mu_m|_{\Sigma} &=\left| \sum_{i=1}^n \Proj( X_i',\tilde R_{m-1},\tilde \mu_{m-1},\Sigma)/n+\cN_d\left(0,\lambda_{m,n}^2\Sigma \right)\right|_{\Sigma}\\
&\leq \Bigg| \sum_{i=1}^n [\Proj( X_i',\tilde R_{m-1},\tilde \mu_{m-1}, \Sigma)- \Proj( X_i,\tilde R_{m-1},\tilde \mu_{m-1},\Sigma)]/n\Bigg|_{\Sigma}\\
&\hspace{6em}+ \Bigg|\sum_{i=1}^n [\Proj( X_i,\tilde R_{m-1},\tilde \mu_{m-1}, \Sigma)-\E(\Proj( X,\tilde R_{m-1},\tilde \mu_{m-1}, \Sigma))]/n\Bigg|_{\Sigma}\\
&\hspace{6em} + \Bigg|\E(\Proj(X,\tilde R_{m-1},\tilde \mu_{m-1}, \Sigma))\Bigg|_{\Sigma} + \Bigg|\cN_d\left(0,\lambda_{m,n}^2\Sigma \right)\Bigg|_{\Sigma}\\
&\coloneqq |I_1|_{\Sigma}+|I_2|_{\Sigma}+|I_3|_{\Sigma}+|I_{4}|_{\Sigma}.
\end{align*}
Here, $I_1$ is the difference between the uncontaminated and contaminated clipped sample mean, $I_2$ is the deviation between the uncontaminated clipped sample mean and its corresponding theoretical mean, $I_3$ is the deviation between the theoretical clipped mean and the true mean, and $I_4$ is magnitude of the noise added to the clipped mean. 
The rest of the proof proceeds by bounding each of $I_1$--$I_4$ above with high probability.

\paragraph{Bounding $I_1$:}  
Note that $|\tilde \mu_{m-1}|_{\Sigma}\leq 3\tilde R_0/4$ by assumption. 
Thus, the conditions of Lemma~\ref{lemma::radius_error} are satisfied. 
Let $\tilde \xi_{q,c}=\inf\{R\colon \tilde\nu(B_{R,\Sigma}(c)^c)\leq q\}$.  
Applying Lemma~\ref{lemma::radius_error} yields that 
\begin{equation}\label{eqn::radius-bound}
\xi_{2\tau',\tilde \mu_{m-1}}\leq \tilde \xi_{3\tau'/2,\tilde \mu_{m-1}}\leq \tilde R_{m-1} \leq  \tilde \xi_{\tau'/2,\tilde \mu_{m-1}}\leq   \xi_{\tau'/4,\tilde \mu_{m-1}},
\end{equation}
with probability at least $1-2\alpha/(3M^*)$. 
Call the event on which \eqref{eqn::radius-bound} holds $E_2$ and assume $E_2$ holds for the remainder of the proof. 

Next, because $\tau'<1\implies \tau'/4<1/4<1/2$ and Condition~\ref{cond::F-inc} holds we can apply  Lemma~\ref{lemma::quantile bound}. 
Applying Lemma~\ref{lemma::quantile bound} we have for any $c\in\rdd$,
\begin{equation}\label{eqn::qb}
    \xi_{\tau'/4,c}\lesssim  \sqrt{\frac{d}{\tau' }}+|c|_{\Sigma}.
\end{equation}
Equation \eqref{eqn::radius-bound}, \eqref{eqn::qb} and the fact that $\eta\leq \tau'/8$ yields that
\begin{multline*}
        |I_1|_{\Sigma} =\left|\sum_{i=1}^n \left[\Proj( X_i',\tilde R_{m-1},\tilde \mu_{m-1},\Sigma)- \Proj( X_i,\tilde R_{m-1},\tilde \mu_{m-1}, \Sigma)\right]/n\right|_{\Sigma}\\ \lesssim \eta\tilde R_{m-1} \lesssim \eta \xi_{\tau'/4,\tilde \mu_{m-1}}
       \lesssim \sqrt{d\eta }+\eta\norm{\tilde \mu_{m-1}}_\Sigma.
\end{multline*}
\paragraph{Bounding $I_2$:} 
Let $Y_i=\Sigma^{-1/2}X_i$ and $\hat\mu_{m-1}=\Sigma^{-1/2}\tilde \mu_{m-1}$ and define $\Proj(a,b,c,I)=\Proj(a,b,c)$ for appropriate $a,b,c$. 
Using this notation, simple algebra yields that
\begin{align*}
    |I_2|_{\Sigma}&=\sup_{
|u|=1 }\sum_{i=1}^n \left(\frac{\left\langle \Proj(Y_i,\tilde R_{m-1},\hat\mu_{m-1}),u\right\rangle-\left\langle\E\Proj(Y_i,\tilde R_{m-1},\hat\mu_{m-1}),u\right\rangle}{n}\right).
\end{align*}

Letting $R_1=\inf_{|c|_{\Sigma}\leq\norm{\tilde \mu_{m-1}}_\Sigma} \xi_{2\tau',c}$ and $R_2=\sup_{|c|_{\Sigma}\leq\norm{\tilde \mu_{m-1}}_\Sigma}\xi_{\tau'/4,c}$, we can write
\begin{multline*}
     n|I_2|_{\Sigma}\leq  \left|\sum_{i=1}^n \Proj(Y_i,(\tilde R_{m-1}-|\hat\mu_{m-1}|)_+,0)\right|
  \\ +\sup_{\substack{
|c|\leq\norm{\hat\mu_{m-1}},\\
R_1\le R\le R_2
}}\Big|\sum_{i=1}^n (\Proj(Y_i,R,c)-\Proj(Y_i,(R-|c|)_+,0)  -\E[\Proj(Y_i,R,c)])\Big|
\coloneqq nI_{2,1}+nI_{2,2}.
\end{multline*}
Here, one notes that $\E\Proj(Y_i,(\tilde R_{m-1}-|\hat\mu_{m-1}|)_+,0)=0$ by symmetry. 
Now, we handle the two terms separately. 

Focusing on $I_{2,1}$, we have $\tilde R_{m-1}-|\hat\mu_{m-1}|\leq K\sqrt{d/\tau'}\coloneqq R_3$ for some $K>0$. In addition we have that Condition~\ref{cond::F-inc} implies that each $Y_i=U_i\ell_i$, where $U_i$ is uniform on the unit sphere and $\ell_i$ is a positive scalar. 
Hence,
\begin{align*}
    I_{2,1}&\leq \sup_{ 
0\le R\le R_3
}\frac{1}{n}\left|\sum_{i=1}^n \Proj(Y_i,R,0)\right|=\sup_{ 
0\le R\le R_3
}\sup_{\norm{u}=1}\frac{1}{n}\sum_{i=1}^n (R\wedge \ell_i)\ip{U_i,u}.
\end{align*}

Now, we condition on $\ell=(\ell_1,\ldots,\ell_n)'$. 
First, observe that
\begin{align*}
     \var\left(\frac{1}{n}\sum_{i=1}^n (R\wedge \ell_i)\ip{U_i,u}|\ell\right)=\frac{1}{n^2}\sum_{i=1}^n  (R\wedge \ell_i)^2 \var\left(\ip{U_i,u}|\ell\right) \leq \frac{1}{n^2d}\sum_{i=1}^n  (R_3\wedge \ell_i)^2\coloneqq\sigma^2  .
\end{align*}
We now define the process
$$Z_{R,u}=\frac{1}{n}\sum_{i=1}^n (R\wedge \ell_i)\ip{U_i,u},$$
where we have conditioned on $\ell=(\ell_1,\ldots,\ell_n)'$. 
Next, we use Dudley’s integral inequality, see \citep{versh}, Theorem 8.1.6, combined with the observation that $Z_{R,u}=Z_{R,u}-Z_{0,u}$. 

First, we have that for admissible $s,t$, it holds that
$\norm{Z_{t,u}-Z_{s,v}}_{L_2}\leq 2\sigma.$
Next, we have that the set $T$ in the notation of \citet{versh} is given by $S_1\times [0,R_3],$ where $S_1=\{u\colon \norm{u}=1\}$. 
We must compute the entropy of $T$ with respect to the $L_2$-metric and show that the process has subgaussian increments. 
To this end, observe that letting $r_i=(R\wedge \ell_i)u-(s\wedge \ell_i)v$, we have that using the fact that $U_i$ are i.i.d. uniform on $S_1$,
\begin{multline*}
    \norm{Z_{R,u}-Z_{s,v}}^2_{L_2}=\frac{1}{n^2}\E\left[\left|\sum_{i=1}^n (R\wedge \ell_i)\ip{U_i,u}-(s\wedge \ell_i)\ip{U_i,v}\right|^2|\ell\right]
    =\frac{1}{n^2}\E\left[\left|\sum_{i=1}^n \ip{U_i,r_i}\right|^2|\ell\right]\\
    =\frac{1}{n^2}\sum_{i=1}^n \E\left[\left|\ip{U_i,r_i}\right|^2|\ell\right]
    =\frac{1}{dn^2}\sum_{i=1}^n \norm{r_i}^2.
\end{multline*}
Now, for $u=v$, we have that
$\norm{r_i}^2\leq |R-s|^2, $
and $    \norm{Z_{R,u}-Z_{s,v}}^2_{L_2}\leq |R-s|^2/(dn).$
For $s=R$, we have that $\norm{Z_{R,u}-Z_{t,v}}_{L_2}\lesssim \sigma\norm{u-v}.$ 
Therefore, we can use these facts, and the usual Euclidean covering numbers, as follows.
$$\log N(\epsilon,T,\norm{\cdot}_{L_2})\lesssim d\log(1+2\sigma/\epsilon)+\log(1+R_3/(\epsilon\sqrt{dn})).$$

The entropy integral is bounded as follows, 
$$\int_0^{2\sigma }\sqrt{\log N(\epsilon,T,\norm{\cdot}_{L_2})}\lesssim \sigma\left[\sqrt{d}+\sqrt{\log\left(e+\frac{R_3}{2\sigma \sqrt{dn}}\right)}\right].$$
Now, it remains to show subgaussianity of the increments. To that end, we have that because the $U_i$ are independent and zero mean,
$$\norm{Z_{R,u}-Z_{t,v}}_{\psi_2}\lesssim n^{-1}\left(\sum_{i=1}^n \norm{\ip{U_i,r_i}}_{\psi_2}^2\right)^{1/2}\lesssim n^{-1}d^{-1/2}\left(\sum_{i=1}^n \norm{r_i}^2\right)^{1/2}.$$

We can now apply Dudley’s integral inequality, see \citep{versh}, Theorem 8.1.6. 
We have with probability $1-\alpha/(2M^*)$, it holds that 
$$I_{2,1}\leq \sup_{(R,u)\in T}\norm{Z_{R,u}-Z_{0,u}}\lesssim \sigma\left[\sqrt{d}+\sqrt{\log\left(e+\frac{R_3}{2\sigma \sqrt{dn}}\right)}\right]+\sqrt{\log(4M^*/\alpha)}\sigma.$$
Simplifying, we have that 
\begin{align*}
   I_{2,1}&\lesssim  \sqrt{\sum_{i=1}^n  (R_3\wedge \ell_i)^2}\left[\frac{1}{n}+\sqrt{\frac{\log\left(e+\frac{R_3}{2\sigma \sqrt{dn}}\right)}{n^2d}}+\sqrt{\frac{\log(4M^*/\alpha)}{n^2d}}\right].
\end{align*}
Now, we bound the coefficient $$\sqrt{\sum_{i=1}^n  (R_3\wedge \ell_i)^2},$$
with high probability. 

We have that by definition, ${\sum_{i=1}^n  \E(R_3\wedge \ell_i)^2}\leq n(d\wedge R_3^2)$. 
In addition, note that $(R_3\wedge \ell_i)^4\leq R_3^2(R_3\wedge \ell_i)^2$, and so $\var( (R_3\wedge \ell_i)^2)\leq  R_3^2 \E(R_3\wedge \ell_i)^2 \leq d R_3^2$. 
Applying Bernstein's inequality yields that with probability at least $1-\alpha/(2M^*)$, we have 
$$\sum_{i=1}^n  [(R_3\wedge \ell_i)^2-\E(R_3\wedge \ell_i)^2]\lesssim R_3\sqrt{nd\log(4M^*/\alpha)}+R_3^2\log(4M^*/\alpha).$$
Consequently, with probability at least $1-\alpha/(2M^*)$, we have 
$$\sum_{i=1}^n  (R_3\wedge \ell_i)^2\lesssim n(d\wedge R_3^2)+R_3\sqrt{nd\log(2M^*/\alpha)}+R_3^2\log(4M^*/\alpha).$$

Now, we have that $R_3\lesssim \sqrt{\frac{d}{\tau'}}$, and $\tau'\gtrsim \log(4M^*/\alpha)/n$. 
In this case, we have that $R_3^2\log(4M^*/\alpha)\lesssim dn$, and $R_3\sqrt{nd\log(4M^*/\alpha)}\lesssim dn$. 
Therefore, we have that with probability at least $1-\alpha/(2M^*)$, we have  
$$\sum_{i=1}^n  (R_3\wedge \ell_i)^2\lesssim nd.$$

Thus, 
$$I_{2,1}\lesssim \sqrt{\frac{d}{n}}+\sqrt{\frac{\log\left(e+\frac{R_3}{2\sqrt{d}}\right)}{n}}+\sqrt{\frac{\log(4M^*/\alpha)}{n}}.$$
Simplifying further, using the condition that $n\geq d$ and $\tau'\gtrsim d/n$, we get that 
$$I_{2,1}\lesssim \sqrt{\frac{d}{n}}+\sqrt{\frac{\log\left(e+\frac{1}{2 \sqrt{\tau'}}\right)}{n}}+\sqrt{\frac{\log(4M^*/\alpha)}{n}} \lesssim \sqrt{\frac{d\vee \log (en/d)}{n}}+\sqrt{\frac{\log(4M^*/\alpha)}{n}}.$$

We now bound $I_{2,2}$. To this end, recall
\begin{align*}
    I_{2,2}&\leq \sup_{\substack{
|c|\leq\norm{\hat \mu_{m-1}},\\
R_1\le R\le R_2
}}\Big|\sum_{i=1}^n [\Proj(Y_i,R,c)-\Proj(Y_i,(R-|c|)_+,0)  -\E[\Proj(Y_i,R,c)]]\Big|/n.
\end{align*}
Because $\E\Proj(Y_i,(R-|c|)_+,0) =0$ for all $R,c$, this is a supremum of a centered empirical process. 
It has a random envelope $\propto \norm{\tilde \mu_{m-1}}_\Sigma$. 
To bound this term with high probability, we use a shell argument, combined with a uniform concentration inequality for the supremum of an empirical process, namely Theorem 5.11 of \citet{vanDeGeer2000}. The trick is that the envelope of this class does not depend on the size of the maximal projection radius $R_2$. The application is long, and requires computing covering numbers and an entropy integral, so we prove it in a series of lemmas at the end of this section. 
Using this technique, whose final result is given in Lemma~\ref{lemma::dyadic-shell}, we get that for universal constants $C_0,C_1,C_3>0$ and $C_2(w)>0$, taking $\alpha/(2M^*)>C_0\log ne^{-C_1n}$, which holds by assumption, we have that if
$$n\geq C_2(w)d\log n,$$
which also holds by assumption,
then with probability at least $1-\alpha/(2M^*)$,
$$I_{2,2}\lesssim 1/n+\norm{\tilde \mu_{m-1}}_\Sigma\left(\sqrt{\frac{d\log n}{n}}  
  +\sqrt{\frac{\log(2M^*C_3 \log n/\alpha)}{n}}\right).   $$
Thus, overall, with probability $1-3\alpha/(2M^*)$, we have that
$$I_2\lesssim (1+\norm{\tilde \mu_{m-1}}_\Sigma)\left(\sqrt{\frac{d\log(en)}{n}}  
  +\sqrt{\frac{\log((2C_3\vee 4)M^*\log n /\alpha)}{n}}\right).$$

\paragraph{Bounding $I_3$:} 
Directly from Lemma~\ref{lemma::proj_e_va_cs}, in combination with \eqref{eqn::radius-bound}, we have
\begin{align*}
    |I_3|_{\Sigma}&\leq \left( 1+\norm{\tilde \mu_{m-1}}_\Sigma\right)\sqrt{\Prr{|X-\tilde \mu_{m-1}|_{\Sigma}>\tilde R_{m-1}}}\lesssim \left(1+\norm{\tilde \mu_{m-1}}_\Sigma\right)\sqrt{\tau'}.
\end{align*}
\paragraph{Bounding $I_4$} Now, by standard subgaussian tail bounds, we have that for $u\geq \sqrt{d}$,
\begin{align*}
\Prr{\norm{Z}>u}&=\Prr{\chi^2_d>u^2}\leq 2\exp(-u^2/10),
\end{align*}
where $Z\sim\cN_d(0,I)$. 
This implies that with probability at least $1-2\exp(-d/10)$, it holds that $\norm{Z}\lesssim \sqrt{d}.$
with probability at least $1-2\exp(-d/10)$, we then have that 
\begin{equation*}
    |I_4|_{\Sigma} =|\cN_d\bigl(0,\lambda_{m,n}^2\times \Sigma\bigr)|_{\Sigma}\lesssim \xi_{\tau'/4,\tilde \mu_{m-1}} \sqrt{\frac{d}{n^2\rho_{\mathrm{mean},m}}}
     \lesssim(\sqrt{\frac{d}{\tau' }}+\norm{\tilde \mu_{m-1}}_\Sigma)\sqrt{\frac{d}{n^2\rho_{\mathrm{mean},m}}}.
\end{equation*}

\paragraph{Putting the bounds together:}
Combining the bounds on each of $I_1$--$I_4$, we get that with probability at least $1-5\alpha/2M^* -2\exp(-d/10)>1-9\alpha/2M^* $ that
\begin{multline*}
        |\tilde \mu_m|_{\Sigma}\lesssim \sqrt{d\eta}+  (1+\norm{\tilde \mu_{m-1}}_\Sigma)\left(\sqrt{\frac{d\log(en)}{n}}  
  +\sqrt{\frac{\log((C_3\vee 2)M^*\log n /\alpha)}{n}}\right)
    \\ +\left(1+\norm{\tilde \mu_{m-1}}_\Sigma\right)\sqrt{\tau'} + (\sqrt{\frac{d}{\tau' }} +\norm{\tilde \mu_{m-1}}_\Sigma)\sqrt{\frac{d}{n^2\rho_{\mathrm{mean},m}}}+\norm{\tilde \mu_{m-1}}_\Sigma\eta.
\end{multline*}
Simplifying with the bound $\tau'\geq d/n$, noting that $\sqrt{\tau'}>\eta$, and plugging in \eqref{eqn::tau} yields 
\begin{multline*}
      |\tilde \mu_m|_{\Sigma}\lesssim  \sqrt{d\eta} +\left(1+\norm{\tilde \mu_{m-1}}_\Sigma\right)\times\\  \Bigg( \sqrt{\frac{d}{n\rho_{\mathrm{mean},m}}}+\sqrt{\frac{ d\log(en)\vee \log((3(h(\beta)+1)\vee (C_3\vee 2)\log n )M^*/\alpha)}{n}}
    +\sqrt{\eta}\Bigg).    \qedhere
\end{multline*}
\end{proof}
\subsection{A helpful bound on the population radius}
\begin{lemma}\label{lemma::quantile bound}
For $0<a<1/2$ and $b\in\rdd$ 
we have that
\begin{equation}\label{eqn::quantile-general}
   \xi_{a,b}\leq \sqrt{\frac{d}{a}}+|b-\mu|_{\Sigma}.
\end{equation}
\end{lemma}
\begin{proof}
First, assume that for $c\in\rdd$, and any $0<q<1/2$ it holds that $\xi_{q,c} > |c-\mu|_{\Sigma}.$ 
Then,
\begin{align*}
    \Prr{\norm{X-b}_{\Sigma}\geq \xi_{a,b}}=\Prr{X\notin B_{\xi_{a,b},\Sigma}(b)}\geq a.
\end{align*}
Now, for $\xi_{a,b}>|b|_{\Sigma}$, applying Markov's inequality yields
\begin{align*}
    a&\leq \Prr{\norm{X-b}_{\Sigma}\geq  \xi_{a,b}}\leq \Prr{\norm{X}_{\Sigma}\geq  \xi_{a,b}-|b|_{\Sigma}}\leq  \frac{d}{(\xi_{a,b}-|b|_{\Sigma})^2}.
\end{align*}
Rearranging yields that
\begin{equation*}\label{eqn::quantile-elliptical}
\xi_{a,b}\leq \sqrt{\frac{d}{a}}+|b|_{\Sigma}. \qedhere
\end{equation*}

It remains to show that for any $c\in\rdd$, and any $0<q<1/2$ it holds that $\xi_{q,c} > |c-\mu|_{\Sigma}.$ 
Without loss of generality, take $\mu=0$ and let $X\sim \nu$. 
Define $\langle x,y\rangle_{\Sigma} = \langle \Sigma^{-1/2}x,\,\Sigma^{-1/2}y\rangle.$ 
If $x\in B_{|c|_{\Sigma},\Sigma}(c)$, then $|x-c|_{\Sigma}^2\leq |c|_{\Sigma}^2$. 
Thus, $ |x|_{\Sigma}^2 + |c|_{\Sigma}^2 - 2\langle x,c\rangle_{\Sigma}
\leq|c|_{\Sigma}^2.$ Canceling $|c|_{\Sigma}^2$ yields
\begin{equation}\label{eqn::innerp-bound}
|x|_{\Sigma}^2 - 2\langle x,c\rangle_{\Sigma} \leq 0,
\qquad\text{so}\qquad
\langle x,c\rangle_{\Sigma} \geq  0.
\end{equation}
Define $- B_{|c|_{\Sigma},\Sigma}(c) = \{-z : z\in  B_{|c|_{\Sigma},\Sigma}(c)\}$.  
If $y\in - B_{|c|_{\Sigma},\Sigma}(c)$, then $y=-x$ for some $x\in  B_{|c|_{\Sigma},\Sigma}(c)$, and immediately $\langle x,c\rangle_{\Sigma} \geq  0\implies \langle y,c\rangle_{\Sigma} \leq  0.$ 
Further, \eqref{eqn::innerp-bound} implies that  $\langle x,c\rangle_{\Sigma} =  0\iff |x|_{\Sigma}=0$. 
 
It follows that $ B_{|c|_{\Sigma},\Sigma}(c)\cap - B_{|c|_{\Sigma},\Sigma}(c) = \{0\}$.  
Central symmetry of $X$, i.e., $X\eqd -X$ implies $\Pr(X\in B_{|c|_{\Sigma},\Sigma}(c)) = \Pr(X\in -B_{|c|_{\Sigma},\Sigma}(c)).$
Now, using $ B_{|c|_{\Sigma},\Sigma}(c)\cap - B_{|c|_{\Sigma},\Sigma}(c) = \{0\}$,
\begin{equation*}
            \Pr(X\in B_{|c|_{\Sigma},\Sigma}(c)\cup -B_{|c|_{\Sigma},\Sigma}(c))
= \Pr(X\in B_{|c|_{\Sigma},\Sigma}(c)) + \Pr(X\in -B_{|c|_{\Sigma},\Sigma}(c))
= 2\,\Pr(X\in B_{|c|_{\Sigma},\Sigma}(c))
\le 1.
\end{equation*}

Thus
\[
\Pr(X\in B_{|c|_{\Sigma},\Sigma}(c)) \le \frac12,
\qquad
\Pr(|X-c|_{\Sigma} \ge |c|_{\Sigma})
= 1 - \Pr(X\in B_{|c|_{\Sigma},\Sigma}(c))
\ge \frac12.
\]
Now, $ \Pr(|X-c|_{\Sigma} > s)$ is clearly nonincreasing in $s$.
From above,
\[
\Pr(|X-c|_{\Sigma} > |c|_{\Sigma}) \ge \frac12 > q.
\]
Then, the smallest radius $s$ for which $\Pr(|X-c|_{\Sigma} > s)< q$ must satisfy $s\ge |c|_{\Sigma}$, or,
$\xi_{q,c} > |c|_{\Sigma}.$
\end{proof}

\subsection{Concentration of the mean for the first iteration}
\begin{lemma}\label{lemma::init_mean}
Assume that $\mu\in B_{\tilde R_0,\Sigma}(\mu_0)$. Then for all $M^*e^{-d/10}<\alpha<M^*/3$, if $n\gtrsim\log(2M^*/\alpha)$ then with probability at least $1-3\alpha/M^*$
\begin{equation*}
\norm{\mu-\tilde \mu_{1}}_{\Sigma}\lesssim  \tilde R_0\Bigg(S\left(\tilde R_0-|\mu_0-\mu|_{\Sigma}\right)
 +\sqrt{\frac{d}{n^2\rho_{\mathrm{mean},1}}}   +\eta+ \sqrt{\frac{\log(2M^*/\alpha)}{n}}\Bigg).
\end{equation*}
\end{lemma}
\begin{proof}
Assume that $\mu=0$ without loss of generality and for brevity, let $\rho=\rho_{\mathrm{mean},1}$. 
We first bound the noise attributed to the Gaussian mechanism. 
By standard subgaussian tail bounds, we have that for $u\geq \sqrt{d}$,
\begin{align*}
\Prr{\norm{Z}>u}&=\Prr{\chi^2_d>u^2}\leq 2\exp(-u^2/10),
\end{align*}
where $Z\sim\cN_d(0,I)$. 
This implies that with probability at least $1-2\exp(-d/10)$, it holds that $\norm{Z}\lesssim \sqrt{d}.$
with probability at least $1-2\exp(-d/10)$, we then have that 
\begin{align*}
   |\cN_d\bigl(0,\lambda_{1,n}^2\times \Sigma\bigr)|_{\Sigma}\lesssim \tilde R_0 \sqrt{\frac{d}{n^2\rho}}.
\end{align*}
Now let $\hat\mu_1=\frac{1}{n}\sum_{i=1}^n\Proj(X_i,\tilde R_0,\mu_0,\Sigma)$ and $\hat\mu_1'=\frac{1}{n}\sum_{i=1}^n\Proj(X_i',\tilde R_0,\mu_0,\Sigma)$. Observe that $|\hat\mu_1-\hat\mu_1'|_{\Sigma}\leq 2 \eta\tilde R_0$. 
With this in mind, we have that with probability at least $1-\alpha/M^*-2\exp(-d/10)$
\begin{align*}
       | \tilde\mu_1|_{\Sigma} &\lesssim |\E(\Proj(X,\tilde R_0,\mu_0,\Sigma))|_{\Sigma}+|\hat\mu_1'-\E(\Proj(X,\tilde R_0,\mu_0,\Sigma))|_{\Sigma}+\widetilde R_0 \sqrt{\frac{d}{n^2\rho}}\\
         &\lesssim |\E(\Proj(X,\tilde R_0,\mu_0,\Sigma))|_{\Sigma}+2 \eta\tilde R_0+|\hat\mu_1-\E(\Proj(X,\tilde R_0,\mu_0,\Sigma))|_{\Sigma}+\widetilde R_0 \sqrt{\frac{d}{n^2\rho}}\\
        &\lesssim  \tilde R_0S\left(\tilde R_0-|\mu_0|_{\Sigma}\right)  +\eta \tilde R_0+ \sqrt{\frac{\log(2M^*/\alpha)(d\wedge (\tilde R_0)^2)}{n}}+\tilde R_0\frac{\log(2M^*/\alpha)}{n}+\widetilde R_0 \sqrt{\frac{d}{n^2\rho}},
\end{align*}
where the last line applies Lemma~\ref{lemma::proj_e_va2} and Lemma~\ref{lemma::pme}.  
\end{proof}
\subsection{Upper bounds on the expected value of the projections}
\begin{lemma}\label{lemma::proj_e_va_cs}
    For $X\sim \EC(0,\Sigma,F)$, we have that for $R>0$ and $c\in\rdd$, we have that
    \begin{equation*}
         |\E\Proj(X,R,c,\Sigma)|_{\Sigma}\leq \sqrt{\Prr{|X-c|_{\Sigma}> R}}+|c|_{\Sigma}\Prr{|X-c|_{\Sigma}> R} \leq (1+|c|_{\Sigma})\sqrt{\Prr{|X-c|_{\Sigma}> R}}.
    \end{equation*}
\end{lemma}
\begin{proof}
We first use the fact that for $Y=\Sigma^{-1/2}X$, and $a=\Sigma^{-1/2}c$, we have that $|\E\Proj(X,R,c,\Sigma)|_\Sigma=|\E\Proj(Y,R,a,I)|\coloneqq|\E\Proj(Y,R,a)|$. 
If $a=0$, then by spherical symmetry the expectation is zero, so we assume $a\neq 0$.

Before continuing, we show the following fact:
\begin{center}
    The vector $\E\Proj(Y,R,a)$ is parallel to $a$.
\end{center}
Let $Q$ be any orthogonal matrix satisfying $Qa=a$. We know such matrices exist, since any rotation about $a$ satisfies these conditions. 
Then, simple algebra yields $Q\Proj(Y,R,a)=\Proj(QY,R,a)$. Furthermore, since $Q$ is orthogonal and $Y$ is spherically symmetric, we have that $QY\eqd Y$. 
Therefore, using these two facts, $ Q\E\Proj(Y,R,a)=\E \Proj(QY,R,a)=\E \Proj(Y,R,a)$. 
Therefore, $\E \Proj(Y,R,a)$ is a vector $v$, such that for any $Q$ such that $Qa=a$, $v$ satisfies $Qv=v$. 
It follows that $v$ must run parallel to $a$, or we can find a rotation $Q$ of $v$ such that $Qv\neq v$. 

Let $u=a/\norm{a}$. Then, since $\E\Proj(Y,R,a)$ is parallel to $a$,
\begin{equation}\label{eqn::ipf}
    |\E\Proj(Y,R,a)|=|\ip{\E\Proj(Y,R,a),u}|.
\end{equation}
In addition, given $\E Y=0$, we have that $\E\Proj(Y,R,a)=\E[\Proj(Y,R,a)-Y]$. 
Now, $\Proj(Y,R,a)-Y=0$ if $|Y-a|\leq R$, and simple algebra yields $$\Proj(Y,R,a)-Y=\left(\frac{R}{\norm{Y-a}}-1\right)(Y-a).$$ 
As a result, it holds that 
$$\E\Proj(Y,R,a)=\E\left(\left(\frac{R}{\norm{Y-a}}-1\right)(Y-a)\ind{|Y-a|> R})\right).$$
Combining this with the inner product formula \eqref{eqn::ipf} above, we have that 
$$|\E\Proj(Y,R,a)|=|\ip{\E\left(\left(\frac{R}{\norm{Y-a}}-1\right)(Y-a)\ind{|Y-a|> R})\right),u}|.$$
Now, crucially, 
$$\left|\frac{R}{\norm{Y-a}}-1\right|\leq 1,$$
and so we can write 
$$|\E\Proj(Y,R,a)|\leq \E[|\ip{Y-a,u}|\ind{|Y-a|> R}].$$

It remains to simplify this term. 
Applying the triangle inequality, followed by the Cauchy-Schwarz inequality and the fact that spherical symmetry gives $\E\ip{Y,u}^2=1$, yields that 
\begin{align*}
  \E[|\ip{Y-a,u}|\ind{|Y-a|> R}]&\leq \E|\ind{|Y-a|> R}\ip{Y,u}|+|a|\Prr{|Y-a|> R}\\
    &\leq \sqrt{\E\ip{Y,u}^2}\sqrt{\Prr{|Y-a|> R}}+|a|\Prr{|Y-a|> R}\\
    &\leq \sqrt{\Prr{|Y-a|> R}}+|a|\Prr{|Y-a|> R}.
\end{align*}
This completes the proof. 
\end{proof}
\begin{lemma}\label{lemma::proj_e_va2}
    For $X\sim \EC(0,\Sigma,F)$, $R>0$ and $c\in\rdd$, we have that 
$$|\E(\Proj(X,R,c,\Sigma))|_{\Sigma}\leq(R+|c|_{\Sigma})S\left((R-|c|_{\Sigma})_+\right).$$
\end{lemma}
\begin{proof}
Let $A_1=\{x\colon |x|_{\Sigma}\leq (R-|c|_{\Sigma})_+\}$. 
By symmetry of $X$ about 0, $\E(\Proj(X,R,c,\Sigma)\ind{X\in A_1})=0$. 
Therefore
\begin{align*}
|\E \Proj(X,R,c,\Sigma)|_{\Sigma}&=\Bigg|\E(\Proj(X,R,c,\Sigma)\ind{X\notin A_1})\Bigg|_{\Sigma}\leq (R+|c|_{\Sigma})S\left((R-|c|_{\Sigma})_+\right),
\end{align*}
which completes the proof.
\end{proof}
\subsection{Concentration of the projected observations}
\begin{lemma}\label{lemma::pme}
If $X\sim \EC(\mu,\Sigma,F)$, and, given $R>0$ and $c\in \rdd$, $\mu\in B_{R,\Sigma}(c)$ then for all $0<\alpha <1$
with probability at least $1-\alpha$, it holds that 
\begin{equation*}
        \left|\frac{1}{n}\sum_{i=1}^n [\Proj(X_i,R,c,\Sigma)-\E(\Proj(X_i,R,c,\Sigma))]\right|_{\Sigma}\lesssim\sqrt{\frac{\log(2/\alpha)(d\wedge R^2)}{n}}+R\frac{\log(2/\alpha)}{n}.
\end{equation*}
\end{lemma}
\begin{proof}
Without loss of generality, assume that $\mu=0$. 
First, let $Y_i=\Proj(X_i,R,c,\Sigma)-\E(\Proj(X_i,R,c,\Sigma))$. 
Now, observe that $\E(Y_i)=0$ and $\norm{Y_i}_{\Sigma}^2\leq (\norm{\Proj(X_i,R,c,\Sigma)-c}_{\Sigma}+ \norm{\E(\Proj(X_i,R,c,\Sigma))-c}_{\Sigma})^2\leq 4R^2$. 
In addition, $\Proj$ is 1-Lipschitz in $x$ (see Lemma~\ref{lemma:proj-lip}) and $\mu\in B_{R,\Sigma}(c)$. 
Therefore, we have that $\E\norm{Y_i}_{\Sigma}^2\leq\E\norm{ X_i}_{\Sigma}^2=d.$ 
Using the Bernstein vector inequality, see  Appendix B of \citep{SmaleYao2006}, for all $t>0$, we have that 
\begin{align*}
    \Prr{\norm{\bar Y}_{\Sigma}\geq t}\leq \exp\left(-n\frac{t^2}{ 2(d\wedge (4R^2))+ 4Rt/3}\right).
\end{align*}
which yields the desired result. 
\end{proof}
\subsection{A high probability bound on the radius} 
Here we prove two lemmata that together show that the private radius concentrates. The first is as follows. 
\begin{lemma}\label{lemma::radius_error}
Suppose that Condition~\ref{cond::F-inc}, and Condition~\ref{cond::mean-ball}--\ref{cond::distinct} hold, $|\tilde \mu_{m}|_\Sigma \leq 3\tilde R_0/4$, $\rho_{\mathrm{bal},m} \geq 16/24349$,
and $\tau=\tau(\alpha,M^*)$ as in \eqref{eqn::tau}, 
then given $2e^{-d/10}<\alpha/M^*<3/2$, for $m\in[M-1]$, if $\tau'<1/2$, then we have with probability at least $1-2\alpha/3M^*$, 
\begin{equation*}
\xi_{2\tau',\tilde \mu_{m}}\leq \tilde \xi_{3\tau'/2,\tilde \mu_{m}}\leq \tilde R_m \leq  \tilde \xi_{\tau'/2,\tilde \mu_{m}}\leq   \xi_{\tau'/4,\tilde \mu_{m}}, 
\end{equation*}
and 
\begin{equation*}
    \xi_{2\tau',\tilde \mu_{m}}\leq \tilde \xi_{3\tau'/2,\tilde \mu_{m}}\leq\xi_{5\tau'/4,\tilde \mu_{m}}\leq  \xi_{3\tau'/4,\tilde \mu_{m}}\leq  \tilde \xi_{\tau'/2,\tilde \mu_{m}}\leq   \xi_{\tau'/4,\tilde \mu_{m}}.
\end{equation*}
\end{lemma}
\begin{proof}
Assume without loss of generality that $\mu=0$ and let $\rho=\rho_{\mathrm{bal},m}$ for the remainder of the proof. 
Now, define $E_1$ to be the event such that the following inequalities hold
\begin{equation}
\label{eqn::step_2}
\xi_{2\tau',\tilde \mu_{m}}\leq \tilde \xi_{3\tau'/2,\tilde \mu_{m}}\leq\xi_{5\tau'/4,\tilde \mu_{m}}\leq  \xi_{3\tau'/4,\tilde \mu_{m}}\leq  \tilde \xi_{\tau'/2,\tilde \mu_{m}}\leq   \xi_{\tau'/4,\tilde \mu_{m}}.
\end{equation}
Equation \eqref{eqn::step_2} combined with the following inequality
\begin{align}\label{eqn::ineq-1}
    &\tilde \xi_{3\tau'/2,\tilde \mu_{m}}\leq \tilde R_m\leq \tilde \xi_{\tau'/2,\tilde \mu_{m}},
\end{align}
yields the final result. 
It remains to prove $E_1$ holds with high probability and, on $E_1$, \eqref{eqn::ineq-1} holds with high probability. 
We start with the latter. 

For $r>0$ and $c\in\rdd$, let $\tilde G(r,c)=\tilde\nu(B_{r,\Sigma}(c)^c)$. One may recall that $\tilde\nu$ is the empirical measure of the contaminated sample. 
The inequality \eqref{eqn::ineq-1} is implied by 
$\tau'/2\leq \tilde G(\tilde R_m,\tilde \mu_{m})\leq 3\tau'/2$. 
Equivalently, \eqref{eqn::ineq-1} it holds if $\tilde G(\tilde R_m,\tilde \mu_{m})$ is within $\tau'/2$ of $\tau'$. 
To prove that this inequality holds with high probability, we can apply Lemma \ref{lem::pq-bound-zCDP} taking $\tau'=\tau'$ and $ t = \tau'/2$. 
In order to apply Lemma \ref{lem::pq-bound-zCDP}, two conditions must be satisfied. 
The first is that $R_{\min}\leq \tilde \xi_{3\tau'/2,\tilde \mu_{m}}$. 
Condition~\ref{cond::beta} part (b) yields that on $E_1$, $R_{\min}\leq \inf_{c\in\rdd}\tilde \xi_{3\tau'/2,c}\leq \tilde \xi_{3\tau'/2,\tilde \mu_{m}}$ on $E_1$. 
The second condition needed to apply Lemma \ref{lem::pq-bound-zCDP} is
\begin{align}
\label{eqn::c1_E1}
    1<\beta&\leq \frac{\tilde \xi_{\tau'/2,\tilde \mu_{m}}-R_{\min}+1}{\tilde \xi_{3\tau'/2,\tilde \mu_{m}}-R_{\min}+1}.
\end{align}
Note that on $E_1$, we have 
\begin{equation*}
         \frac{\tilde \xi_{\tau'/2,\tilde \mu_{m}}-R_{\min}+1}{\tilde \xi_{3\tau'/2,\tilde \mu_{m}}-R_{\min}+1} =1+\frac{\tilde \xi_{\tau'/2,\tilde \mu_{m}}-\tilde \xi_{3\tau'/2,\tilde \mu_{m}}}{\tilde \xi_{3\tau'/2,\tilde \mu_{m}}-R_{\min}+1} >1+\inf_{|c|_{\Sigma}\leq |\tilde \mu_{m}|_\Sigma}\frac{\xi_{3\tau'/4,c}-\xi_{5\tau'/4,c}}{\xi_{5\tau'/4,c}-R_{\min}+1}\geq 1+b_n. 
\end{equation*}
Condition \ref{cond::beta} part (a) directly says that $\beta<1+b_n$, yielding \eqref{eqn::c1_E1}. 

We can now apply Lemma \ref{lem::pq-bound-zCDP}, which yields that 
\begin{multline*}
    \Prr{\left|\tau' - \tilde\nu(B_{\tilde R_m,\Sigma}(\tilde \mu_m)^c)\right| > \tau'/2|E_1}\leq \left(
\frac{\log\!\bigl(\tilde \xi_{\tau'/2,\widetilde \mu_m}+1-R_{\min}\bigr)}{\log \beta}
+ 1
\right)   \\
\times \exp\!\left(
-\frac{n ((\tau'-2/n) \vee (\tau'-2/n)^2)}{16}\,(\sqrt{\rho} \vee \rho)
\right).
\end{multline*}
We simplify the right-hand side.  
Observe that since $|\tilde\mu_{m}|_{\Sigma}\leq 3\tilde R_0/4$ and Condition~\ref{cond::beta} part (b) holds, we have from Lemma~\ref{lemma::quantile bound}, and the fact that $\tau'>16384d/n>16d/n$: 
\begin{equation*}
    \xi_{\tau'/4,\tilde\mu_{m}}\leq 2\sqrt{\frac{d}{\tau'}}+|\tilde\mu_{m}|_{\Sigma}\leq \sqrt{n}/2+|\tilde\mu_{m}|_{\Sigma}\leq \sqrt{n}/2+3\tilde R_0/4\leq \tilde R_0. 
\end{equation*}  
This fact, combined with \eqref{eqn::step_2} yields that $\tilde \xi_{\tau'/2,\tilde \mu_{m}}<\xi_{\tau'/4,\tilde \mu_{m}}<\tilde R_0 $ on $E_1$. 
It follows that on $E_1$,
\begin{align*}
\frac{\log( \tilde \xi_{\tau'/2,\tilde \mu_{m}} + 1-R_{\min})}{\log\beta}&\leq \frac{\log( \tilde R_0 + 1-R_{\min})}{\log\beta}=h(\beta).
\end{align*}
Thus, noting that $(\tau'-2/n)>(\tau'-2/n)^2$, we have
\begin{align*}
     \Prr{\left|\tau' - \tilde\nu(B_{\tilde R_m,\Sigma}(\tilde \mu_{m})^c)\right| > \tau'/2|E_1}
&\leq \left(
h(\beta)+ 1
\right)   
 \exp\!\left(
-\frac{n (\tau'-2/n) }{16}\,(\sqrt{\rho} \vee \rho)
\right).
\end{align*}
Now, using the fact that $\rho\geq 8^2\times 1/986^2=16/24349$ and $$\tau'\geq \frac{988}{n}\log\left(\left[3\frac{M^*}{\alpha}\left[h(\beta)+ 1\right]\right]  \vee e\right) \geq \frac{986}{n}\log\left(\left[3\frac{M^*}{\alpha}\left[h(\beta)+ 1\right]\right]  \vee e\right)+\frac{2}{n},$$
we have that 
\begin{equation*}
    \Prr{\left|\tau' - \tilde\nu(B_{\tilde R_m,\Sigma}(\tilde \mu_m)^c)\right| > \tau'/2|E_1}\leq \frac{\alpha}{3M^*}.
\end{equation*}
Equivalently, on $E_1$, we have that \eqref{eqn::ineq-1} holds with probability at least $ 1-\alpha/3M^*.$

The next step is to show that $E_1$ occurs with high probability. In other words, the inequalities \eqref{eqn::step_2} hold simultaneously with high probability. 
Now, we have that the inequalities \eqref{eqn::step_2} are implied by the two events
\begin{align*}  
    \bigcap_{|c|_{\Sigma}\leq |\tilde \mu_{m}|_\Sigma}\left\{\xi_{2\tau',c}\leq \tilde \xi_{3\tau'/2,c}\leq  \xi_{5\tau'/4,c}\right\} \text{ and }
    \bigcap_{|c|_{\Sigma}\leq |\tilde \mu_{m}|_\Sigma}\left\{ \xi_{3\tau'/4,c}\leq \tilde \xi_{\tau'/2,c}\leq  \xi_{\tau'/4,c}\right\}.
\end{align*} 
As a result, we aim to bound the probability of the following events above: 
\begin{align*}  
    \bigcup_{|c|_{\Sigma}\leq |\tilde \mu_{m}|_\Sigma}\left\{\tilde \xi_{3\tau'/2,c}\leq \xi_{2\tau',c} \right\}
   &\text{ and } \bigcup_{|c|_{\Sigma}\leq |\tilde \mu_{m}|_\Sigma}\left\{ \tilde \xi_{3\tau'/2,c}\geq  \xi_{5\tau'/4,c}\right\},\\
   \bigcup_{|c|_{\Sigma}\leq |\tilde \mu_{m}|_\Sigma}\left\{ \tilde \xi_{\tau'/2,c}\leq \xi_{3\tau'/4,c} \right\} &\text{ and }
   \bigcup_{|c|_{\Sigma}\leq |\tilde \mu_{m}|_\Sigma}\left\{ \tilde \xi_{\tau'/2,c}\geq  \xi_{\tau'/4,c}\right\} .
\end{align*}
above. 

Define $\hat\nu$ to be the empirical measure of the uncontaminated sample and define $\hat G(r,c)=\hat\nu(B_{r,\Sigma}(c)^c)$ for $r>0$ and $c\in\rdd$. 
We next relate $\tilde G$ to $\hat G$, where recall $\tilde  G(r,c)=\tilde\nu(B_{r,\Sigma}(c)^c)$ for $r>0$ and $c\in\rdd$. 
Observe that by the definition of $\tau'$, $\sup_{r\in\re,|c|_{\Sigma}\leq |\tilde \mu_{m}|_\Sigma}|\tilde G(r,c)-\hat G(r,c)|\leq \eta\leq \tau'/8$. 
It follows that 
\begin{align*}
    \tilde \xi_{3\tau'/2,c}\leq \xi_{2\tau',c}&\implies 3\tau'/2\geq \tilde G(\xi_{2\tau',c},c)\\
    &\implies 3\tau'/2\geq \hat G(\xi_{2\tau',c},c)-\tau'/8 \\
    &\iff 3\tau'/2+\tau'/8 \geq \hat G(\xi_{2\tau',c},c)\\
    &\iff (3/2+1/8-2)\tau' \geq \hat G(\xi_{2\tau',c},c)-2\tau'\\
    &\iff 2\tau'- \hat G(\xi_{2\tau',c},c)\geq 3\tau'/8.
\end{align*}
Similarly, we have that 
\begin{align*}
    \tilde \xi_{\tau'/2,c}\geq  \xi_{\tau'/4,c}&\implies \hat G(\xi_{\tau'/4,c},c)- \tau'/4\geq \tau'/8.
\end{align*}
Using the same logic we get that 
$$ \tilde \xi_{3\tau'/2,c}\geq  \xi_{5\tau'/4,c} \implies \hat G(\xi_{5\tau'/4,c},c)- 5\tau'/4\geq \tau'/8,$$
and that 
$$  \xi_{3\tau'/4,c}\leq \tilde \xi_{\tau'/2,c} \implies  3\tau'/4-\hat G( \xi_{3\tau'/4,c},c)\geq \tau'/8.$$
With these final four inequalities in mind, it then suffices to bound 
$$\Prr{\sup_{\substack{|c|_{\Sigma}\le |\tilde \mu_{m}|_\Sigma \\ q\in\{1/4,\,5/4,\,3/2,\,2\}}}|G(\xi_{q\tau',c},c)- \hat G(\xi_{q\tau',c},c)|\geq \tau'/8},$$
where $  G(r,c)=\nu(B_{r,\Sigma}(c)^c)$ for $r>0$ and $c\in\rdd$. 
The above probability can be bounded via empirical process theory. 
The class of functions
$$\cG=\{g(x)=\ind{\norm{x-c}_\Sigma\leq \xi_{q\tau',c}}\colon |c|_{\Sigma}\le |\tilde \mu_{m}|_\Sigma, q\in\{1/4,\,5/4,\,3/2,\,2\}\}$$
is contained in the class of functions $$\cF=\{g(x)=\ind{\norm{x-c}_\Sigma\leq t}\colon c\in\rdd, t\in\re\},$$
which has Vapnik-Chervonenkis (VC) dimension $d+1$. 
Now, noting that for any $|c|_{\Sigma}\le |\tilde \mu_{m}|_\Sigma, q\in\{1/4,\,5/4,\,3/2,\,2\}$ we have that $\var({\ind{|X_i-c|_\Sigma\geq \xi_{q\tau',c}}})\leq 2\tau'$ for $i\in[n]$. 
Therefore, using standard empirical process bounds, (e.g., see \citep{Sen2022} Theorem 7.13), we have that
\begin{align*}
    \E\sup_{\substack{|c|_{\Sigma}\le |\tilde \mu_{m}|_\Sigma \\ q\in\{1/4,\,5/4,\,3/2,\,2\}}}|G(\xi_{q\tau',c},c)- \hat G(\xi_{q\tau',c},c)|&\leq  8\sqrt{\tau' (d+1)/n}.
\end{align*}
Now, observe that by the definition of $\tau'$ 
we have that $\tau'>16384(d+1)/n$, which yields 
$$\E\sup_{\substack{|c|_{\Sigma}\le |\tilde \mu_{m}|_\Sigma \\ q\in\{1/4,\,5/4,\,3/2,\,2\}}}|G(\xi_{q\tau',c},c)- \hat G(\xi_{q\tau',c},c)|\leq \tau'/16.$$ 
We now apply Talagrand's inequality, see Lemma~\ref{lemma::talagrand} in the supplementary material for convenience, with $U\leq 1/n$ and $\sigma^2\leq 2\tau'/n^2$ to get that,
\begin{align*}  
\Prr{\sup_{\substack{|c|_{\Sigma}\le |\tilde \mu_{m}|_\Sigma \\ q\in\{1/4,\,5/4,\,3/2,\,2\}}}|G(\xi_{q\tau',c},c)- \hat G(\xi_{q\tau',c},c)|\geq \tau'/8}&\leq 
\exp\left(\frac{-3n\tau'}{16(96\sqrt{2}+49)}\right) \leq\exp\left(\frac{-n\tau'}{986}\right).
\end{align*}
Recall that $$\tau'\geq 986\frac{\log(3M^*/\alpha)}{n}.$$
Using this fact, we have that $\exp\left(-n\tau'/986\right)\leq \alpha/3M^*$. 
Therefore, we have shown that $E_1$ occurs with probability at least $1-\alpha/3M^*$ and the proof is complete.
\end{proof}

We now prove the second lemma. 
Recall that the corrupted sample $X_1',\ldots,X_n'$ has empirical measure denoted by $\tilde\nu$. We define $\tilde \xi_{q,c}=\inf\{R\colon \tilde\nu(B_{R,\Sigma}(c)^c)\leq q\}$ for $q\in[0,1]$ and $c\in\rdd$.  
\begin{lemma}\label{lem::pq-bound-zCDP}
Assume that Condition~\ref{cond::distinct} holds.  
For all $0<\tau'\leq 1$ and $1/n<t<\tau'\wedge(1-\tau')$, $0<R_{\min}<\tilde \xi_{\tau'+t,\tilde \mu_{m}}$, $m\in[M-1]$, $\rho_{\mathrm{bal},m}>0$ and $1< \beta \leq ({\tilde \xi_{\tau'-t,\tilde \mu_{m}} -R_{\min}+1})/({\tilde \xi_{\tau'+t,\tilde \mu_{m}} -R_{\min}+1}),$ we have that 
\begin{multline}
    \Prr{
\bigl| \tau' - \tilde\nu\!\bigl(B_{\widetilde R_m,\Sigma}(\widetilde \mu_m)^c\bigr) \bigr| > t
}
\leq 
\left(
\frac{\log\!\bigl(\tilde \xi_{\tau'-t,\widetilde \mu_m}+1-R_{\min}\bigr)}{\log \beta}
+ 1
\right) \notag \\
 \times 
\exp\!\left(
-\frac{n ((t-1/n) \vee (t-1/n)^2)}{4}\,(\sqrt{\rho_{\mathrm{bal},m}} \vee \rho_{\mathrm{bal},m})
\right).
\end{multline}
\end{lemma}
\begin{proof}[Proof of Lemma~\ref{lem::pq-bound-zCDP}]
For notational simplicity, we let $\rho=\rho_{\mathrm{bal},m}$ for the remainder of the proof. 
Recall $\hat \tau'=\tau'-V/(n\sqrt{\rho})\eqd \tau'-V/(n\sqrt{\rho})$, where $\eqd$ refers to `is equal in distribution'. 
Observe that 
\begin{equation*}
            \Prr{|\tau' -\tilde\nu(B_{\widetilde R_m,\Sigma}(\tilde \mu_{m})^c)|> t}=\Prr{\tilde\nu(B_{\widetilde R_m,\Sigma}(\tilde \mu_{m})^c)-\tau'> t}    + \Prr{\tau'-\tilde\nu(B_{\widetilde R_m,\Sigma}(\tilde \mu_{m})^c)> t} \coloneqq I+II.
\end{equation*}
First, for $\{|\tau' -\tilde\nu(B_{\widetilde R_m,\Sigma}(\tilde \mu_{m})^c)|\leq t\}$ to be non-empty, we must ensure that there is some $N\in\bbN$ such that $\tilde\nu(B_{\beta^N-1+R_{\min},\Sigma}(\tilde \mu_{m})^c)$ is within $t$ of $\tau'$. 
Let $k$ be the integer at which the process terminates. 
We have that  
\begin{align*}
    \tilde\nu(B_{\widetilde R_m,\Sigma}(\tilde \mu_{m})^c)-\tau'> t
    &\implies      k< \frac{\log(\tilde \xi_{\tau'+t,\tilde \mu_{m}} -R_{\min}+1)}{\log \beta}\coloneqq k_1,\\
    \tau'- \tilde\nu(B_{\widetilde R_m,\Sigma}(\tilde \mu_{m})^c)> t
    &\implies     k> \frac{\log(\tilde \xi_{\tau'-t,\tilde \mu_{m}} -R_{\min}+1)}{\log \beta}\coloneqq k_2.
\end{align*}
Thus, as long as there is an integer between $k_1$ and $k_2$, then $\{|\tau' -\tilde\nu(B_{\widetilde R_m,\Sigma}(\tilde \mu_{m})^c)|\leq t\}$ is non-empty. This is implied by the condition $1 \leq k_2 - k_1$, which is equivalent to
\begin{equation*}
        1 \leq \frac{\log(\tilde \xi_{\tau'-t,\tilde \mu_{m}} -R_{\min}+1)}{\log \beta}- \frac{\log(\tilde \xi_{\tau'+t,\tilde \mu_{m}} -R_{\min}+1)}{\log \beta} 
    \implies \beta \leq \frac{\tilde \xi_{\tau'-t,\tilde \mu_{m}} -R_{\min}+1}{\tilde \xi_{\tau'+t,\tilde \mu_{m}} -R_{\min}+1}. 
\end{equation*}
The above inequality holds by assumption, and so then $\{|\tau' -\tilde\nu(B_{\widetilde R_m,\Sigma}(\tilde \mu_{m})^c)|\leq t\}$ is non-empty.

We now bound $I$ above. 
Then, using the definition of $k_1$ and $k$, the fact that $V$ and $\{V_i\}_{i=1}^\infty$ are standard normal random variables and a union bound, it holds that
\begin{align*}
I\leq \Prr{k<k_1}&\leq \sum_{i=0}^{\floor{k_1}}\Prr{V_{i}/(n\sqrt{\rho})\geq 1-\hat{\tau'}-\tilde\nu(B_{\beta^{i}+R_{\min}-1,\Sigma}(\tilde \mu_{m}))}\\
&\leq  \sum_{i=0}^{\floor{k_1}}\Prr{V_{i}/(n\sqrt{\rho})\geq 1-\tau' + V/(n\sqrt{\rho}) - \tilde\nu(B_{(\tilde \xi_{\tau'+t}-R_{\min}+1)+R_{\min}-1,\Sigma}(\tilde \mu_{m}))}\\
&\leq (k_1+1)\Prr{V_{\floor{k_1}}-V\geq (t-1/n)n\sqrt{\rho}}\\
 &\leq \left(\frac{\log(\tilde \xi_{ \tau'+t,\tilde \mu_{m}}+1-R_{\min})}{\log\beta}+1\right)\exp\left(-n((t-1/n)\vee (t-1/n)^2)\sqrt{\rho\vee \rho^2}/4 \right).
\end{align*}
Here we use an exponential tail bound for the standard Gaussian, as we will be applying this result when $t<1$. We also make use of Condition~\ref{cond::distinct} on line 3. 
Now, we consider $II$. We assume that $k_2$ is an integer here for readability, where we note that the non-integer case results in the same bound. The non-integer case instead uses the first grid index above $k_2$. 
Observe that 
\begin{align*}
II\leq \Prr{k> k_2}&\leq \Prr{\text{Query } k_2 \text{ fails}}\\
&=\Prr{V_{k_2}/(n\sqrt{\rho})+\tilde\nu(B_{\beta^{k_2}+R_{\min}-1,\Sigma}(\tilde \mu_{m}))\leq 1-\hat \tau'}\\
&= \Prr{\frac{V_{k_2}}{n\sqrt{\rho}}+\tilde\nu(B_{\tilde \xi_{\tau'-t}-R_{\min}+1+R_{\min}-1,\Sigma}(\tilde \mu_{m}))\leq 1-\tau'+\frac{V}{n\sqrt{\rho}}}\\
&\leq\Prr{tn\sqrt{\rho} \leq V-V_{k_2}}\\
&\leq \exp\left(-n(t\vee t^2)\sqrt{\rho\vee \rho^2}/4 \right).
\end{align*}
Note that we make use of Condition~\ref{cond::distinct} on line 4. 
\end{proof}
\subsection{Bound on $I_{2}$ via shelling and empirical process techniques}
For a function class $\cF$, $\epsilon>0$ and a norm $\norm{\cdot}$, let $\sH(\epsilon,\cF,\norm{\cdot})$ be the bracketing entropy of $\cF$.
\begin{lemma}\label{lemma:entropy-g}
Given any probability measure $P$ on $\rdd$, the class of functions $\cH=\{h_{R,c,u}(x)=\ip{\Proj(x,R,c)-\Proj(x,(R-|c|)_+,0),u} \colon\ (R,c,u)\in [0,R^*]\times  B_{r} \times  S_1\} $ with $r\leq R^*$ satisfies  
$$\int_0^z\sH^{1/2}(\epsilon,\cH,L_2(P))d\epsilon\lesssim \sqrt{d}z\sqrt{\log(e+8e(2+r+R^*)^2/z)}.$$
\end{lemma} 
\begin{proof}

First note that, given $x\in\rdd$, the map $h_x((R,c,u))=h_{R,c,u}(x)=\ip{\Proj(x,R,c)-\Proj(x,(R-|c|)_+,0),u}$ is  $4\sqrt{2+r^2}$-Lipschitz in $(R,c,u)$ with respect to the Euclidean metric. (For a long form proof of this, see the supplementary material.) 
It follows from Lemma 2.7.11 of \citet{vanDerVaartWellner1996} and standard bounds on the covering number of $d-1$-dimensional balls and spheres that
$$\sH^{1/2}(\alpha ,\cH,L_2(P))\lesssim \sqrt{d\log(1+8(1+r+R^*)\sqrt{2+r^2}/\alpha)}.$$
Let $a=8(1+r+R^*)\sqrt{2+r^2}$. 
Thus, 
\begin{align*}
    \int_0^z\sH^{1/2}(\alpha,\cH,L_2(P))d\alpha&\lesssim \sqrt{d}\int_0^z\sqrt{\log(1+a/\alpha)}d\alpha
\lesssim\sqrt{d}z\sqrt{\log(e+ea/z)}.
\end{align*}
Lastly, we use the fact that $r\leq R^*$ and apply $a\leq 8(2+r+R^*)^2$.
\end{proof}
\begin{lemma}\label{lemma::em-pro-bound}
Given $\cH=\{h(x)=h_{\mathbf{R},c,u}(x)=\ip{\Proj(x,\mathbf{R},c)-\Proj(x,(\mathbf{R}-|c|)_+,0),u} \colon\ (\mathbf{R},c,u)\in [0,R^*]\times  B_{r} \times  S_1\} $ so that $r\leq R^*$, there exists constants $C_0,C_1,C_2,C_3>0$, such that for $C_0e^{-C_1n}<\alpha<1$, if
$n\geq C_2d\log\left(e+(4e(2+r+R^*)^2)/r\right),$ 
then with probability at least $1-\alpha$,
\begin{equation*}
    \sup_{h\in \cH}\frac{1}{n}\left|\sum_{i=1}^n  [h( \Sigma^{-1/2}X_i)-\E h( \Sigma^{-1/2}X_i)]\right|\lesssim r\sqrt{\frac{d\log\left(e+\frac{4e(2+r+R^*)^2}{r}\right)\vee \log(C_3/\alpha)}{n}} .  
\end{equation*}
\end{lemma}
\begin{proof}
We are working with a supremum of an empirical process, by which we denote
$$\zeta_n=\sup_{h\in \cH}\left|\sum_{i=1}^n  [h( \Sigma^{-1/2}X_i)/n-\E(h( \Sigma^{-1/2}X_i))/n]\right|.$$ 
We apply Theorem 5.11 of \citet{vanDeGeer2000}, with $K$ and $R$ in their notation as follows.  
Take $\{Y_i\}_{i=1}^n\coloneqq \{\Sigma^{-1/2}X_i\}_{i=1}^n$. 
In the notation of \citet{vanDeGeer2000}, we have that $R=2r\geq\sup_{h\in\cH}\sqrt{\E  h( Y_i)^2}) ,$
and 
$K=2r$, so that $\rho_K(h)\leq R$, where $\rho_K$ is an extension of the $L_2(P)$ norm. Here, $P$ is the probability measure associated with $Y_i$, see page 72 of \citet{vanDeGeer2000}. 



Applying Theorem 5.11 of \citet{vanDeGeer2000}, being careful about satisfying the conditions, where here $a$ in their notation is the upper bound on $\zeta_n$ up to a constant, yields that for $Ce^{-C_2n}<\alpha<1$, if
$n\geq C_0^2d\log\left(e+(4e(2+r+R^*)^2)/r\right)/C_1^2,$ then with probability at least $1-\alpha$,
\begin{equation*}
   \zeta_n\lesssim r\left(\sqrt{\frac{d\log\left(e+\frac{4e(2+r+R^*)^2}{r}\right)}{n}}  
  +\sqrt{\frac{\log(C_2/\alpha)}{n}}\right).
\end{equation*}
\end{proof}
\begin{lemma}\label{lemma::dyadic-shell}
Under the conditions of Lemma~\ref{lemma:one-iteration}, there exists universal constants $C_0,C_1,C_3>0$ and $C_2(w)>0$, such that for $C_0e^{-C_1n}\log n<\alpha<1$, if
$$n\geq C_2(w)d\log n,$$
then with probability at least $1-\alpha$,
\begin{equation}\label{eqn::shellbound}
    I_{2,2}\lesssim 1/n+\norm{\tilde \mu_{m-1}}_\Sigma\left(\sqrt{\frac{d\log n}{n}}  
  +\sqrt{\frac{\log(C_3 \log n/\alpha)}{n}}\right).   
\end{equation}
\end{lemma}
\begin{proof}
We need to handle the randomness in $\norm{\tilde \mu_{m-1}}_\Sigma$, for which we use a peeling argument. Let $c_m=\norm{\tilde \mu_{m-1}}_\Sigma$. We know for sure that $c_m\leq 3n^w/4$ by Condition~\ref{cond::beta} and the assumptions of the theorem. 
We consider the following set of shell events $\cS_{j}=\{2^{j-1}/n<c_m\leq 2^j/n \}$, for $j=1,\ldots, \ceil{(w+1)\log_2 n}\coloneqq 1,\ldots, J$. 
We have that either exactly one $\cS_{j}$ has to occur or $\norm{\tilde \mu_{m-1}}_\Sigma\leq 1/n$. 
If $\norm{\tilde \mu_{m-1}}_\Sigma\leq 1/n$, then result holds immediately. 

We turn to the case that $\norm{\tilde \mu_{m-1}}_\Sigma> 1/n$. 
We define $\cH_j=\{h(x)=h_{\mathbf{R},c,u}(x)=\ip{\Proj(x,\mathbf{R},c)-\Proj(x,(\mathbf{R}-|c|)_+,0),u} \colon\ (\mathbf{R},c,u)\in [0,R^*]\times  B_{r_j} \times  S_1\} $, so that $r_j=2^j/n\leq 2n^w=R^*$. 
Now, we let $E_j$ be the event that the desired bound holds over $\cH_j$, that is, 
for $$\zeta_{n,j}=\sup_{h\in \cH_j}\frac{1}{n}\left|\sum_{i=1}^n  [h( \Sigma^{-1/2}X_i)-\E(h( \Sigma^{-1/2}X_i))]\right|,$$
we define
\begin{equation*}
       E_j =  \Bigg\{   \zeta_{n,j} \leq K r_j\left(\sqrt{\frac{d\log n}{n}}  
  +\sqrt{\frac{\log(C_3 \log n/\alpha)}{n}}\right)  \Bigg\},
\end{equation*}
for some large enough constant $K>0$. 
Now, if $\cS_{j}$ occurs, then the random function class that is present in $I_{2,2}$ is contained in $\cH_j$, and so $I_{2,2}\leq \zeta_{n,j}$. Here, one notes that $R_2< R^*$. 

Now, note that on $\cS_{j}$, $r_j/2<c_m\leq r_j\implies r_j<2c_m$, so then $E_j$ would imply that 
$$E=  \Bigg\{I_{2,2}\leq 2Kc_m\left(\sqrt{\frac{d\log n}{n}}  
  +\sqrt{\frac{\log(C_3 \log n/\alpha)}{n}}\right) \Bigg\},$$
occurs. 
Thus, $\Prr{E}\geq  \Prr{\bigcap_{j=1}^J E_j}\geq 1-\sum_{j=1}^J  \Prr{ E_j^c}.$ 
Therefore, if $E_j^c$ has probability at most $\alpha/J$, then the result holds. 

To show that each $E_j^c$ has probability at most $\alpha/J$, we apply Lemma~\ref{lemma::em-pro-bound} to each shell with failure probability $\alpha/J$, $r=2^j/n$, and $2n^w=R^*$. 
This yields that if
$$n\geq C_2d\log\left(e+4en(2+4n^w)^2\right),$$
then with probability at least $1-\alpha/J$,
\begin{equation*}
   \zeta_{n,j}\lesssim \frac{2^{j}}{n}\left(\sqrt{\frac{d\log\left(e+4en(2+4n^w)^2\right)}{n}}  
  +\sqrt{\frac{\log(C_3 J/\alpha)}{n}}\right).
\end{equation*}
This bound on $n$ is implied by the lemma assumptions, and the bound on $\zeta_{n,j}$ above is stronger than that of $E_j$. 
This completes the proof. 
\end{proof}
\subsection{Proof of Proposition~\ref{prop::priv-compu}}
\begin{proof}[Proof of Proposition~\ref{prop::priv-compu}]
The privacy guarantee follows directly from composition (Proposition~\ref{prop::dp}), the Gaussian mechanism (Proposition~\ref{prop::adm}), and the privacy guarantee of (zero-concentrated) \texttt{AboveThreshold} for quantiles \citep{Ramsay2025}. The computation time can be proven as follows. The initial covariance matrix inversion and decomposition takes $O(d^3)$ time. Then, we can whiten the data once, taking $O(nd^2)$ time. For each iteration, we compute the relevant Euclidean norm for each data point, $O(nd)$. The radius computation can be implemented in $O(n\log n )$ time \citep{Durfee2024}. The projections take a worst case $O(nd)$ time, and the noise mean also takes $O(nd)$ time. Therefore, the final time is $O(d^3+nd^2+Mnd+Mn\log n)$. 
\end{proof}


\bibliographystyle{cas-model2-names}

\bibliography{main}

\newpage
\section{Supplementary proofs}
We state the following lemma, whose proof is straight-forward and is included in the supplementary material. 
\begin{lemma}\label{lemma:proj-ip-lip-2}
Given $x\in\rdd$, the map $h_x((R,c,u))=\ip{\Proj(x,R,c)-\Proj(x,(R-|c|)_+,0),u}$, where $h_x: [0,R^*]\times B_{r} \times  S_1\to \re $ is  $4\sqrt{2+r^2}$-Lipschitz in $(R,c,u)$ with respect to the Euclidean metric. 
\end{lemma} 
\begin{proof}

We have that 
\begin{multline*}
    h_x((R,c,u))= \\ \hspace{2em}\ip{\Proj(x,R,c)-\Proj(x,R,0),u}+\ip{\Proj(x,R,0)-\Proj(x,(R-|c|)_+,0),u} \\ \coloneqq h_{x,1}((R,c,u))+h_{x,2}((R,c,u)).
\end{multline*}
We first show that $h_{x,1}$ is $2\sqrt{2+r^2}$-Lipschitz in the vector $(R,c,u)$.

We consider each of $R,c,u$ separately, and then combine with the triangle inequality. 
First, note that for any $z\in B_{r}$ and $R_1>R_2>0$, we have that $\ip{\Proj(x,R_1,z),u}-\ip{\Proj(x,R_2,z),u})=0$ if $|x-z|\leq R_2$. Otherwise, note that the projection lies along $x-c$.  the maximum distance between the projected points is $|R_2-R_1|$. 
Accordingly, the Cauchy-Schwarz inequality yields that
\begin{align*}
    |h_{x,1}((R_1,c,u))-h_{x,1}((R_2,c,u))|&\leq 2|R_2-R_1|,
\end{align*}
and the function is 2-Lipschitz in $R$.

Using Lemma~\ref{lemma:proj-lip}, we have that
\begin{align*}
    |h_{x,1}((R,c_1,u))-h_{x,1}((R,c_2,u))|&\leq |\Proj(x,R,c_1)-\Proj(x,R,c_2)|\\ &=|c_1-c_2+\Proj(x-c_1,R,0)-\Proj(x-c_2,R,0)|\leq 2|c_1-c_2|.
\end{align*}
Therefore, the function $h_{x,1}$ is 2-Lipschitz in $c$. 

It remains to consider $u$, from which it is easy to see that
\begin{align*}
    |h_{x,1}((R,c,u))-h_{x,1}((R,c,v))|\leq \norm{\Proj(x-c,R,0)-\Proj(x,R,0)+c}|u-v| &\leq  2r|u-v|.
\end{align*}

Thus, combining these bounds with the Cauchy-Schwarz inequality yields : 
\begin{align*}
    |h_{x,1}((R_2,c_2,u))-h_{x,1}((R_1,c_1,v))|&\leq  |h_{x,1}((R_2,c_1,v))-h_{x,1}((R_1,c_1,v))|\\
    &\hspace{2em} +|h_{x,1}((R_2,c_2,u))-h_{x,1}((R_2,c_2,v))|\\ &\hspace{2em} +|h_{x,1}((R_2,c_2,v))-h_{x,1}((R_2,c_1,v))|\\
    &\leq 2|R_2-R_1|+2|c_1-c_2|+2r|u-v|\\
    &\leq \sqrt{4+4+4r^2}\norm{(R_2,c_2,u)-(R_1,c_1,v)}.
\end{align*}
that the function is $2\sqrt{2+r^2}$-Lipschitz in the vector $(R,c,u)$.

Now, we use similar logic to analyze $h_{x,2}$. 
We have
\begin{multline*}
|h_{x,2}((R_2,c,u))-h_{x,2}((R_1,c,u))|\\
   \leq  |\ip{\Proj(x,R_2,0)-\Proj(x,R_1,0)-\Proj(x,(R_2-|c|)_+,0)+\Proj(x,(R_1-|c|)_+,0),u}| \leq  2|R_2-R_1|,
\end{multline*}
and,
$$|h_{x,2}((R,c_1,u))-h_{x,2}((R,c_2,u))|\leq |\ip{\Proj(x,(R-|c_1|)_+,0)-\Proj(x,(R-|c_2|)_+,0),u}|\leq 2|c_2-c_1|,$$
and,
$$|h_{x,2}((R,c,u))-h_{x,2}((R,c,v))|\leq|\ip{\Proj(x,R,0)-\Proj(x,(R-|c|)_+,0),u-v}|\leq r|u-v|.$$
Therefore, the second term is $\sqrt{8+r^2}$-Lipschitz in the vector $(R,c,u)$. 
Thus, the function is $4\sqrt{2+r^2}$-Lipschitz in the vector $(R,c,u)$.
\end{proof}
\section{Useful existing results}
\begin{lemma}[\cite{Talagrand1994,Sen2022}]\label{lemma::talagrand}
Given a separable class of measurable, real functions $\cF$ such that for $f\in \cF$ we have that $\norm{f}_\infty<U$,  $\sigma^2=\sup_{f\in\cF} \E f^2(X_i)$, and $\E(f(X_i))=0$, it holds that for $t>0$ 
\begin{multline*} 
    \Prr{\sup_{f\in\cF}\sumn f(X_i) \geq  \E \sup_{f\in\cF}\sumn f(X_i)+t}
    \leq \exp\left(\frac{-t^2/2}{2U\E \sup_{f\in\cF}\sumn f(X_i)+n\sigma^2+Ut/3}\right)  \\ \leq\exp\left(\frac{-t^2/2}{2U\sqrt{n\VC(\cF)}+n\sigma^2+Ut/3}\right).
\end{multline*}
\end{lemma}
\begin{proof}
Let $Z=\sup_{f\in\cF}\sumn f(X_i)$, $\sigma^2=\sup_{f\in\cF} \E f^2(X_i)$ and $V=2U\E Z+n\sigma^2$. We have that applying Talagrand's/Bousquet's inequality \citep{Talagrand1994}, it holds that 
    \begin{align*}
        \Prr{Z\geq \E Z+t}\leq \exp\left(\frac{-2t^2}{2V+2Ut/3}\right).
\end{align*}
Now, we have that from (85) of \cite{Sen2022}, it holds that 
$\E Z/n\lesssim \sqrt{\VC(\cF)/n}.$ 
\begin{align*}
        \Prr{Z\geq \E Z+t}\leq \exp\left(\frac{-t^2/2}{2U\sqrt{n\VC(\cF)}+n\sigma^2+Ut/3}\right).&\qedhere
\end{align*}
\end{proof}
\begin{lemma}\label{lemma:proj-lip}
Consider a Hilbert space $\cH$ with norm induced by the inner product $\langle \cdot,\cdot\rangle$. For nonempty, closed, and convex $C\subset \cH$, define the projection of $x\in\cH$ onto $C$ as
\[
\Proj_C(x)\in \argmin_{z\in C}\norm{x-z}.
\]
Then $\Proj_C(x)$ is 1-Lipschitz. 
\end{lemma} 
\begin{proof}
Let $p=\Proj_C(x)$ and $q=\Proj_C(y)$. Given that $C$ is closed and convex, we have that all $z\in C$,
\[
\langle x-p,\, z-p\rangle \le 0
\andd
\langle y-q,\, z-q\rangle \le 0.
\]
Taking $z=q$ in the first inequality and $z=p$ in the second yields
\[
\langle x-p,\, q-p\rangle \le 0
\andd
\langle y-q,\, p-q\rangle \le 0.
\]
Adding them gives
\[
\langle x-p,\, q-p\rangle +
\langle q-y,\, q-p\rangle \le 0 \iff \langle x-y+q-p,\, q-p\rangle \le 0,
\]
which then implies
\[
\langle x-y-(p-q),\, p-q\rangle \ge 0
\iff
\langle x-y,\, p-q\rangle \ge \norm{p-q}^2.
\]
By the Cauchy--Schwarz inequality,
\[
\norm{p-q}^2 \le \langle x-y,\, p-q\rangle \le \norm{x-y}\,\norm{p-q}.
\]
\end{proof}

\section{Simulation results from sensitivity study}\label{app::tuning}
We conducted a sensitivity analysis to assess how the performance of the balloon mean depends on the tuning parameters. 
We then varied one tuning parameter at a time while fixing the others at default values. The default values were given by $\tilde \mu_0=(0,\ldots,0)^\top$, $\tilde R_0=50\sqrt d$, $\beta=1.01$, $M=2$, $\tau_{1}=\tau_2=0.9$ for the high-$\tau$ variant, and $\tau_1=0.8$ and $\tau_2=0.6$ for the low-$\tau$ variant. When assessing each tuning parameter, we varied the tuning parameters as follows, $\tilde R_0=c\sqrt d$ with $c\in\{2,4,5,10,15,30\}$, $\beta\in\{1.001,1.01,1.05,1.1\}$,  $\tau\in\{0.7,0.8,0.85,0.9,0.95,0.99\}$, $M\in\{2,3,4,5,6\}$. When we varied $M$, we set $\tau_{1}=\ldots=\tau_{M-1}=0.9$ for the high-$\tau$ variant, and $\tau_{M-2}=0.8$ and $\tau_{M-1}=0.6$, otherwise $\tau_i=0.9$ for the low-$\tau$ variant. 
For $\tilde \mu_0$, we considered $\tilde\mu_0 = c\mu$, with $c \in \{1.0,1.5,2.0,4.0,10.0\}$.  
This moves $\tilde \mu_0$ along the line running from the origin to $\mu$ and beyond.  
For each parameter, we considered Gaussian, contaminated Gaussian, $t_3$, and skewed distributions, dimensions $d\in\{5,20,50,100\}$, sample sizes $n\in\{500,2000,5000\}$, and privacy budgets $\rho\in\{0.1,1\}$, repeating each configuration for 200 trials. Performance was measured using Mahalanobis error with respect to the true mean. 
The full results are given below. 
\subsection{Effect of initial mean}
We observed minimal sensitivity of the method to the initial mean $\tilde \mu_0$, see below. 
\begin{figure}[htbp]
\centering
  \includegraphics[width=\linewidth]{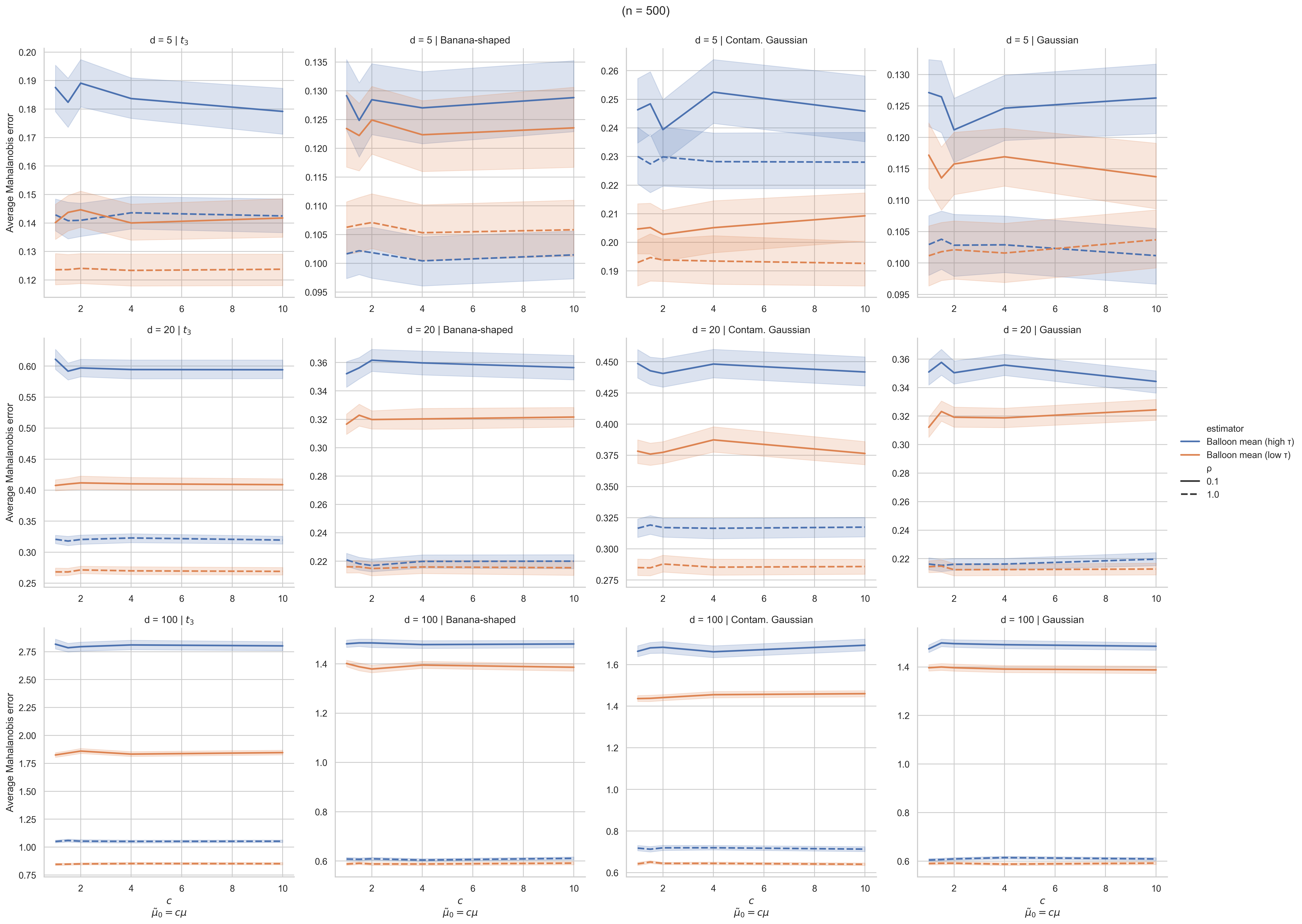}
\caption{Average Mahalanobis error as a function of $c$, where $\tilde\mu_0=c\mu$ under four distributions and differing dimensions $d$ at $n=500$. Solid lines show the mean error over repetitions, with shaded bands indicating a 95\% confidence interval over the simulation runs. We report the balloon mean with high-$\tau$ and low-$\tau$ settings, across privacy levels $\rho \in \{0.1,1.0\}$ (linetypes).}
\label{fig:sweep_mu0_by_n_500}
\end{figure}
\begin{figure}[htbp]
\centering
  \includegraphics[width=\linewidth]{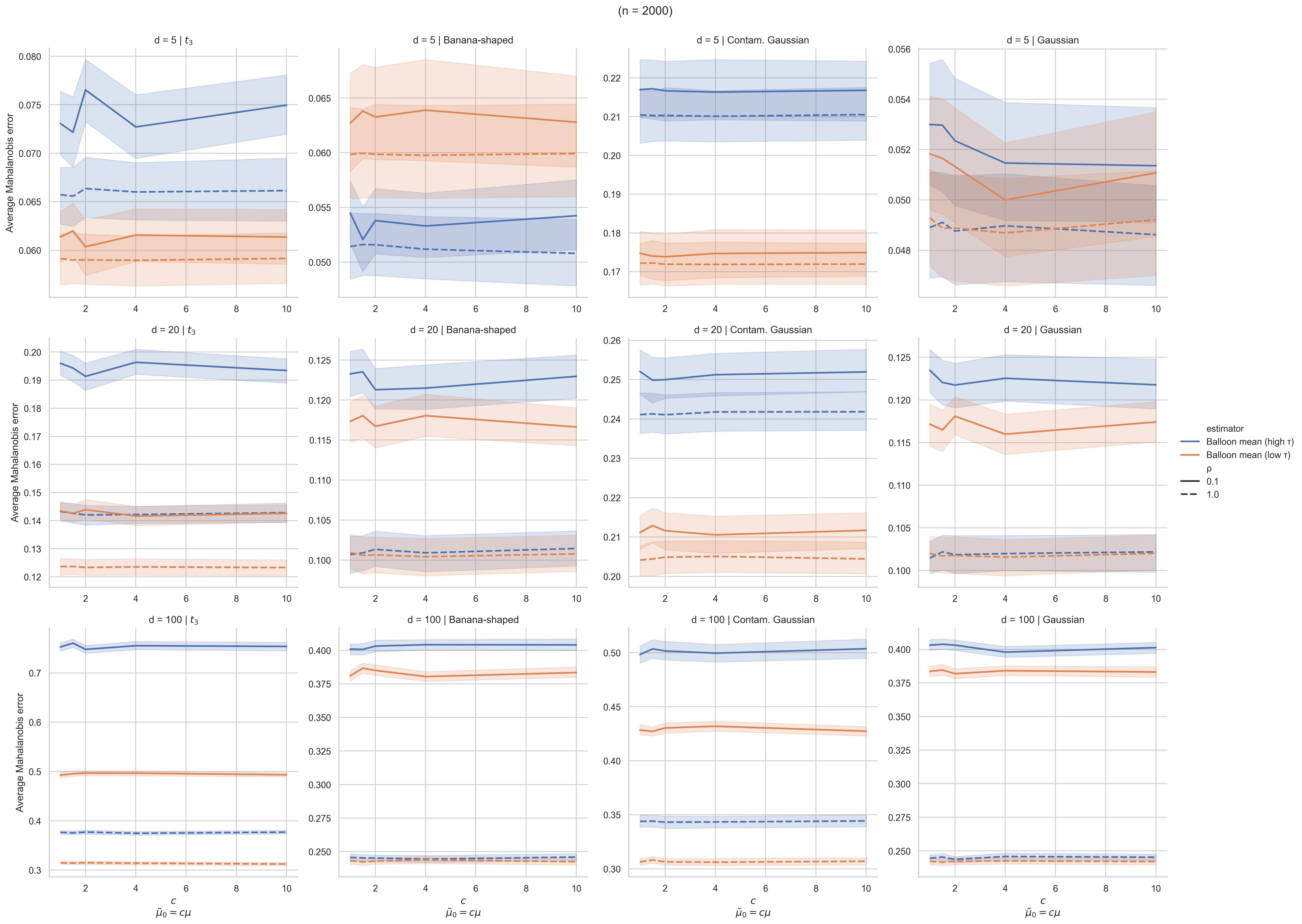}
\caption{Average Mahalanobis error as a function of $c$, where $\tilde\mu_0=c\mu$ under four distributions and differing dimensions $d$ at $n=2000$. Solid lines show the mean error over repetitions, with shaded bands indicating a 95\% confidence interval over the simulation runs. We report the balloon mean with high-$\tau$ and low-$\tau$ settings, across privacy levels $\rho \in \{0.1,1.0\}$ (linetypes).}
\label{fig:sweep_mu0_by_n_2000}
\end{figure}
\FloatBarrier
\subsection{Effect of initial radius}
We observed minimal sensitivity of the method to the radius $\tilde R_0$, see below. 
\begin{figure}[htbp]
\centering
\includegraphics[width=\linewidth]{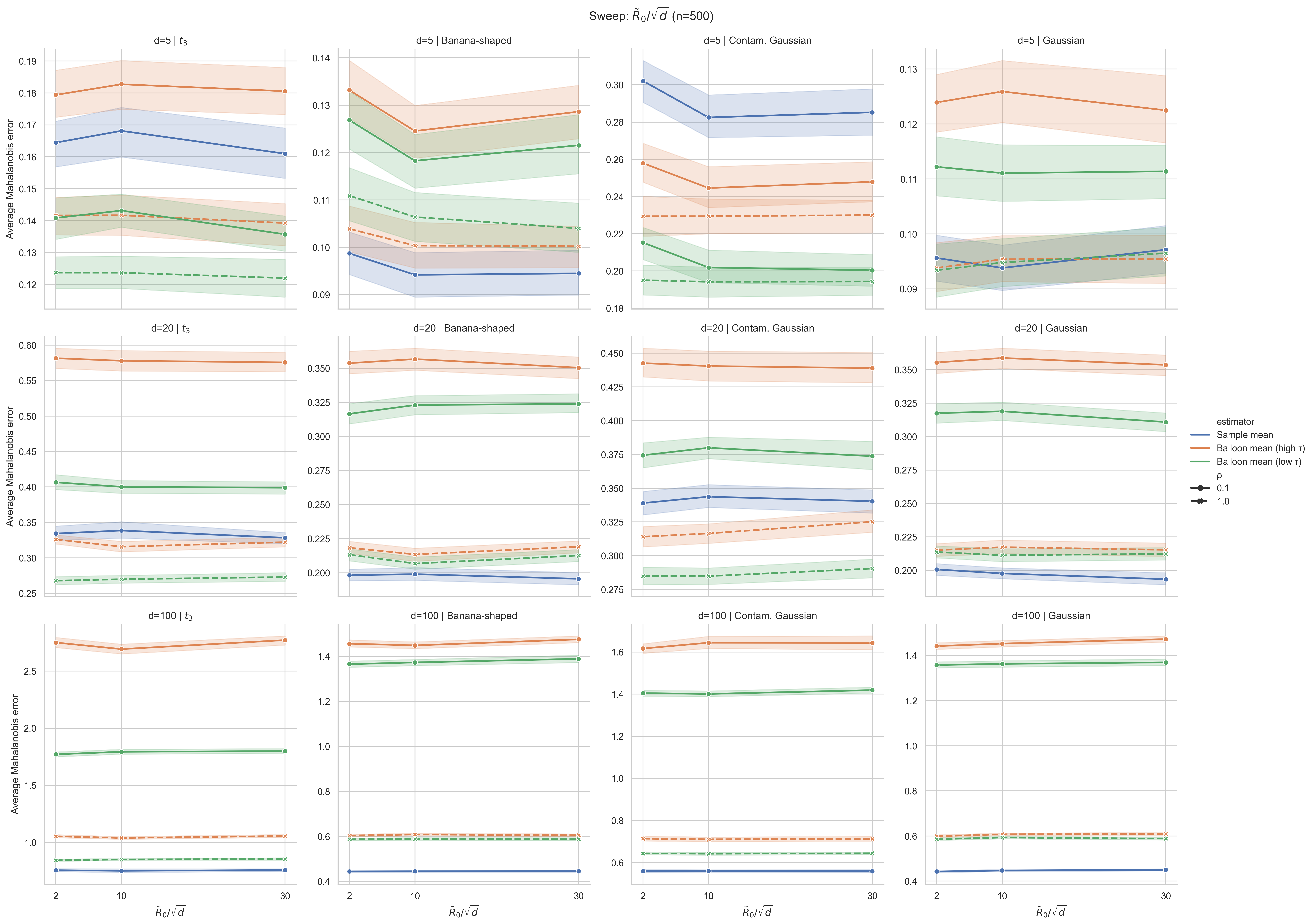}
\caption{Average Mahalanobis error as a function of $\tilde R_0/\sqrt{d}$ for mean estimation under four distributions and differing dimensions $d$ at $n=500$. Solid lines show the mean error over repetitions, with shaded bands indicating a 95\% confidence interval over the simulation runs. We report the balloon mean with high-$\tau$ and low-$\tau$ settings, across privacy levels $\rho \in \{0.1,1.0\}$ (linetypes).}
\label{fig:sweep_Rmax_by_n_500}
\end{figure}
\begin{figure}[htbp]
\centering
\includegraphics[width=\linewidth]{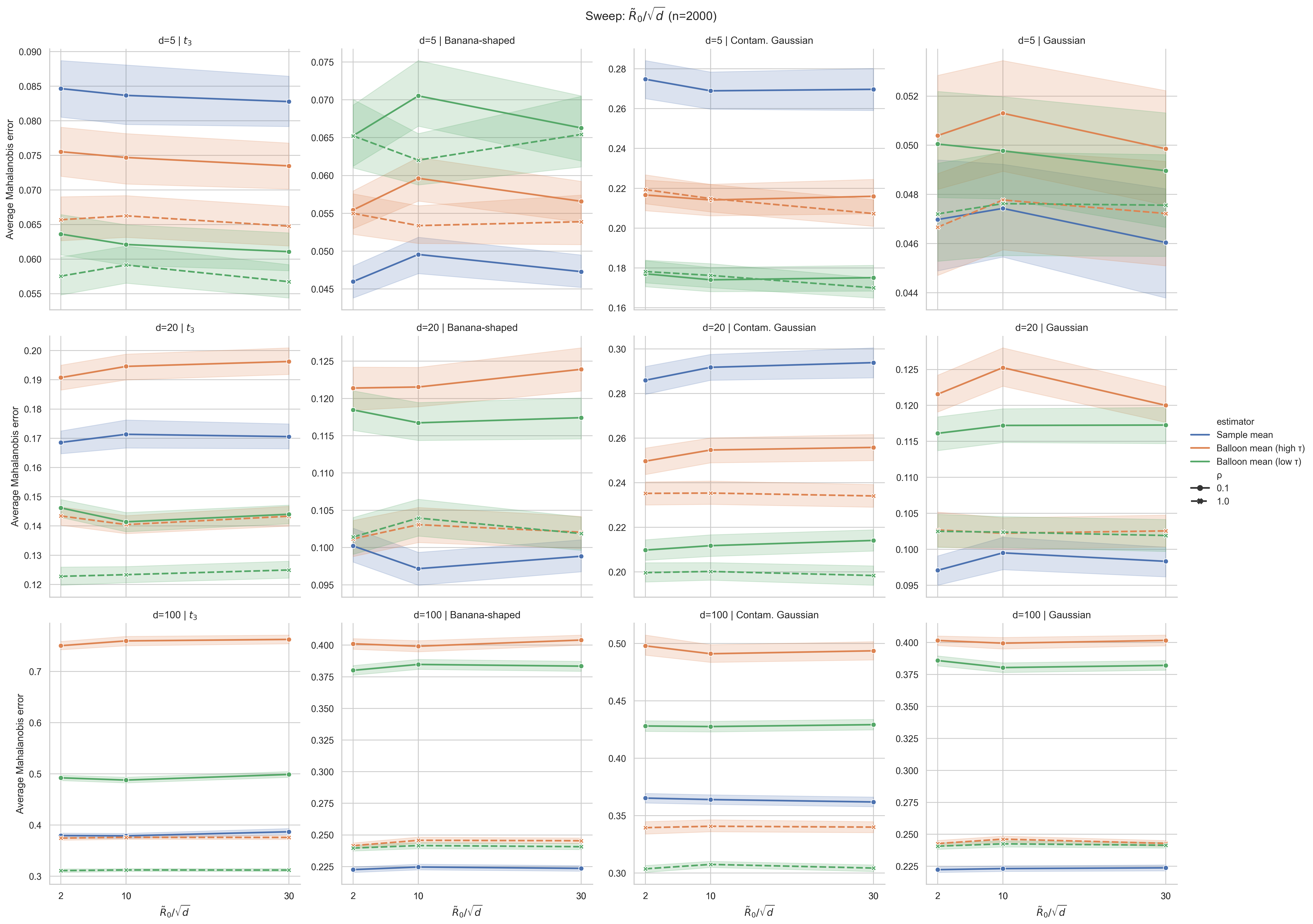}
\caption{Average Mahalanobis error as a function of $\tilde R_0/\sqrt{d}$ for mean estimation under four distributions and differing dimensions $d$ at $n=2000$. Solid lines show the mean error over repetitions, with shaded bands indicating a 95\% confidence interval over the simulation runs. We report the balloon mean with high-$\tau$ and low-$\tau$ settings, across privacy levels $\rho \in \{0.1,1.0\}$ (linetypes).}
\label{fig:sweep_Rmax_by_n_2k}
\end{figure}
\FloatBarrier

\subsection{Effect of grid size}
We observed minimal sensitivity of the method to the grid size $\beta$, see below. 
\begin{figure}[htbp]
\centering
    \includegraphics[width=\linewidth]{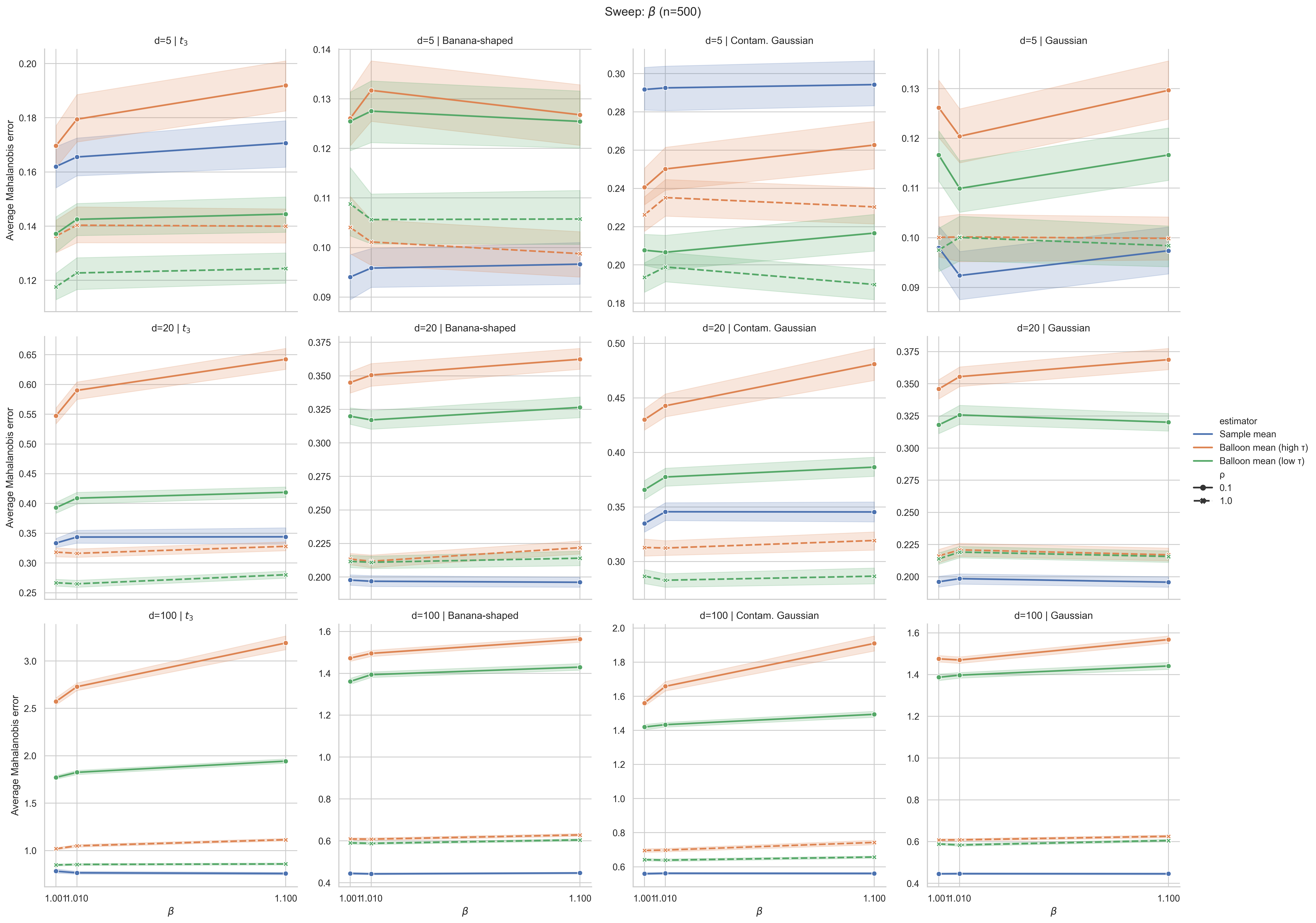}
    \caption{Average Mahalanobis error as a function of $\beta$ for mean estimation under four distributions and differing dimensions $d$ at $n=500$. Solid lines show the mean error over repetitions, with shaded bands indicating a 95\% confidence interval over the simulation runs. We report the balloon mean with high-$\tau$ and low-$\tau$ settings, across privacy levels $\rho \in \{0.1,1.0\}$ (linetypes).}
    \label{fig:beta-500}
\end{figure}
\begin{figure}[htbp]
\centering
    \includegraphics[width=\linewidth]{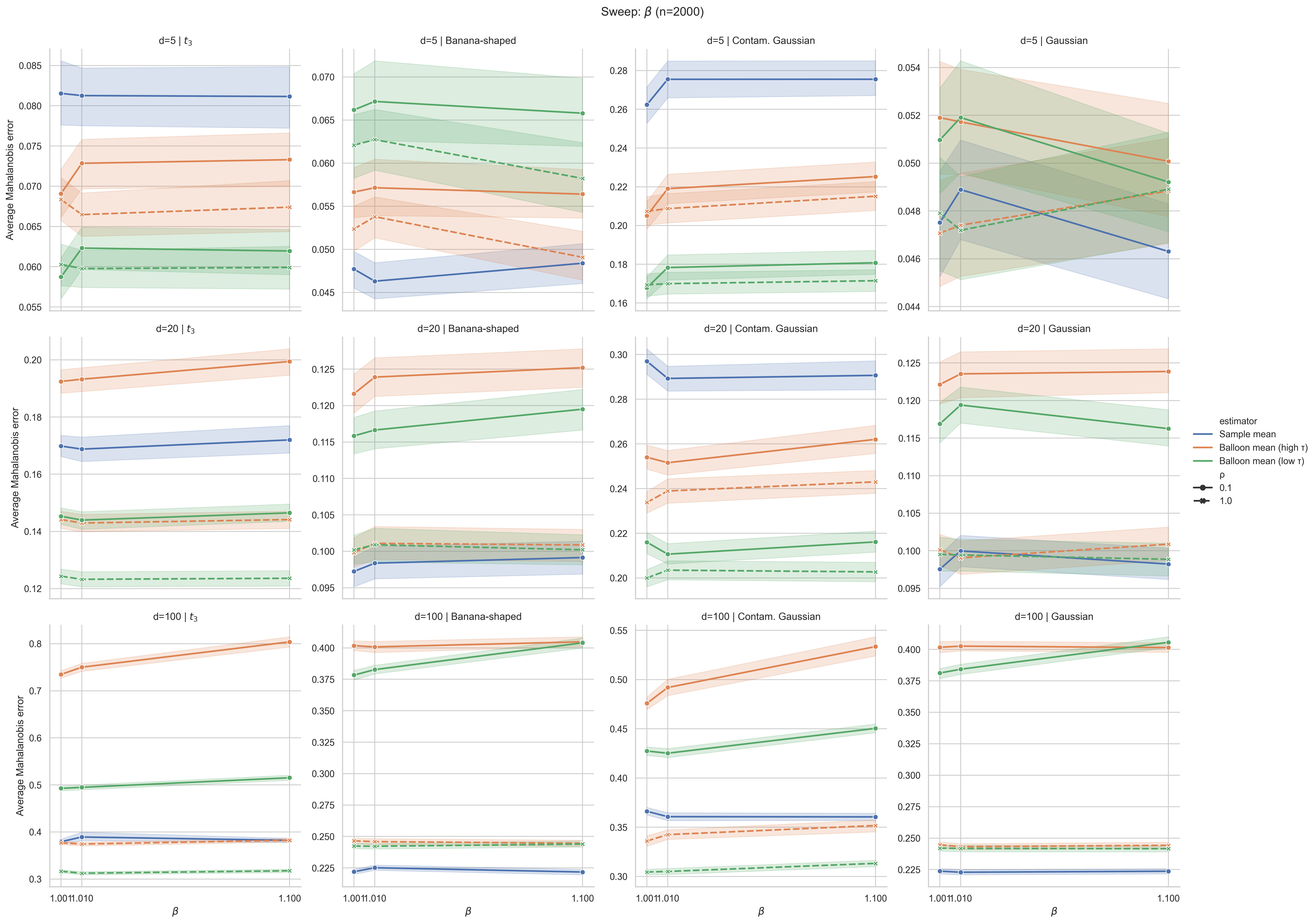}
        \caption{Average Mahalanobis error as a function of $\beta$ for mean estimation under four distributions and differing dimensions $d$ at $n=2000$. Solid lines show the mean error over repetitions, with shaded bands indicating a 95\% confidence interval over the simulation runs. We report the balloon mean with high-$\tau$ and low-$\tau$ settings, across privacy levels $\rho \in \{0.1,1.0\}$ (linetypes).}
    \label{fig:beta-2000}
\end{figure}
\FloatBarrier
\subsection{Effect of the target fraction $\tau$}
Here we present the results for the effect of the target fraction $\tau$. In high-dimensions smaller $\tau$ values generally perform better. This is consistent across generating distribution and $\tau$ variant. In the contaminated settings (heavy-tailed and contaminated Gaussian) smaller $\tau$ values perform better, as expected. Looking at the clean settings in low-moderate dimensions, we first consider the low-$\tau$-variant. When $n=500$, $\tau$ has relatively minimal effect on the low-$\tau$ variant. When $n$ is large, and the dimension is very small, larger $\tau$ appears to be preferred, but this effect disappears when $d$ becomes moderate and again lower $\tau$ values are preferred. Looking at the high-$\tau$ variant, when $n=500$, lower $\tau$ values are preferred. When $n=2000$, $\tau$ has little effect in the Gaussian setting when $d=2$, and smaller $\tau$ is preferred when $d=20$. For the banana-shaped distribution, when $d=2$ larger $\tau$ is preferred, and smaller $\tau$ is preferred when $d=20$. There is a visible effect of $\tau$ , but it is not large; across all settings, varying 
$\tau$ only moderately changes the error.
\begin{figure}[htbp]
\centering
    \includegraphics[width=\linewidth]{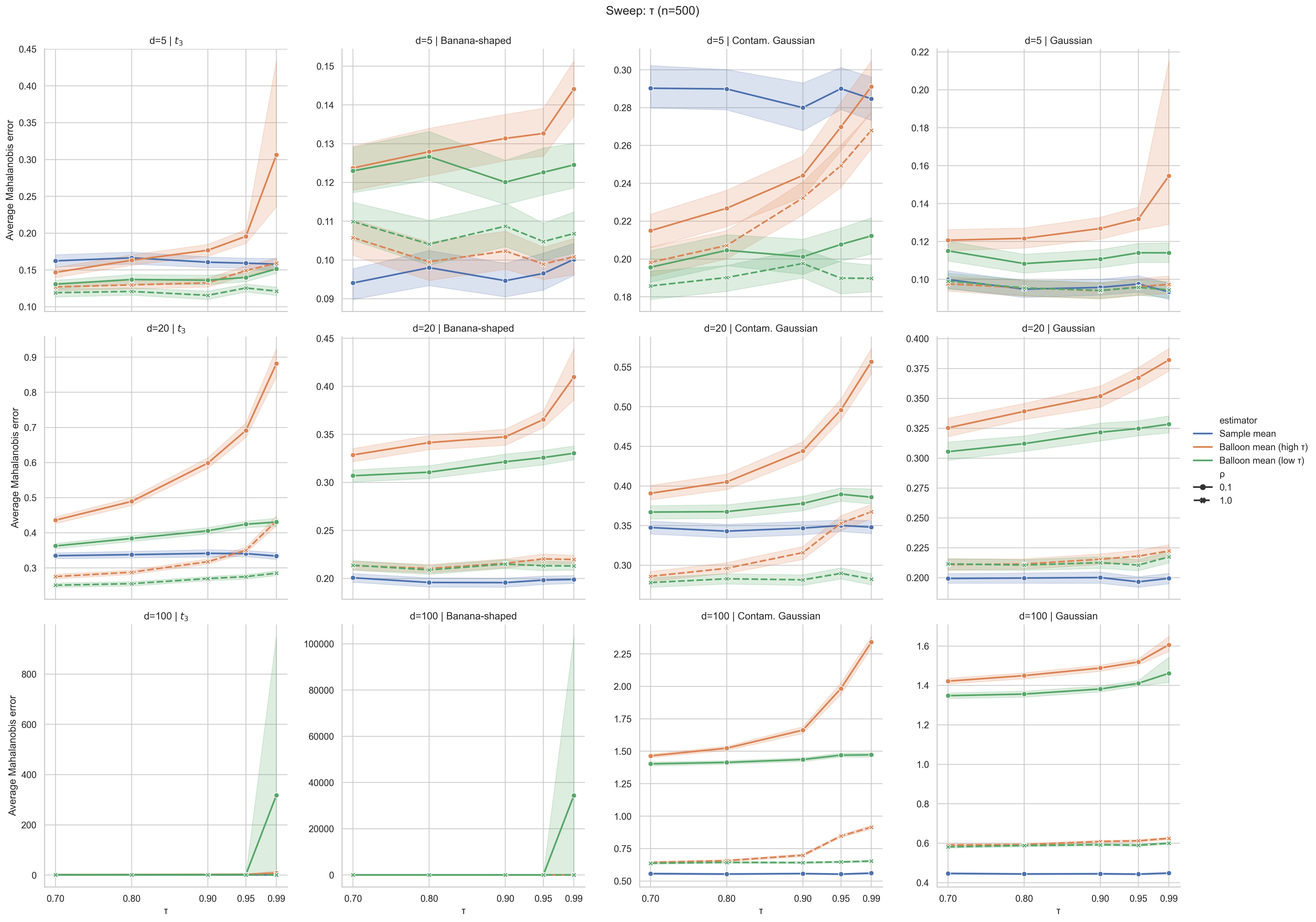}
        \caption{Average Mahalanobis error as a function of $\tau$ for mean estimation under four distributions and differing dimensions $d$ at $n=500$. Solid lines show the mean error over repetitions, with shaded bands indicating a 95\% confidence interval over the simulation runs. We report the balloon mean with high-$\tau$ and low-$\tau$ settings, across privacy levels $\rho \in \{0.1,1.0\}$ (linetypes).}
    \label{fig:tau-500}
\end{figure}
\begin{figure}[htbp]
\centering
    \includegraphics[width=\linewidth]{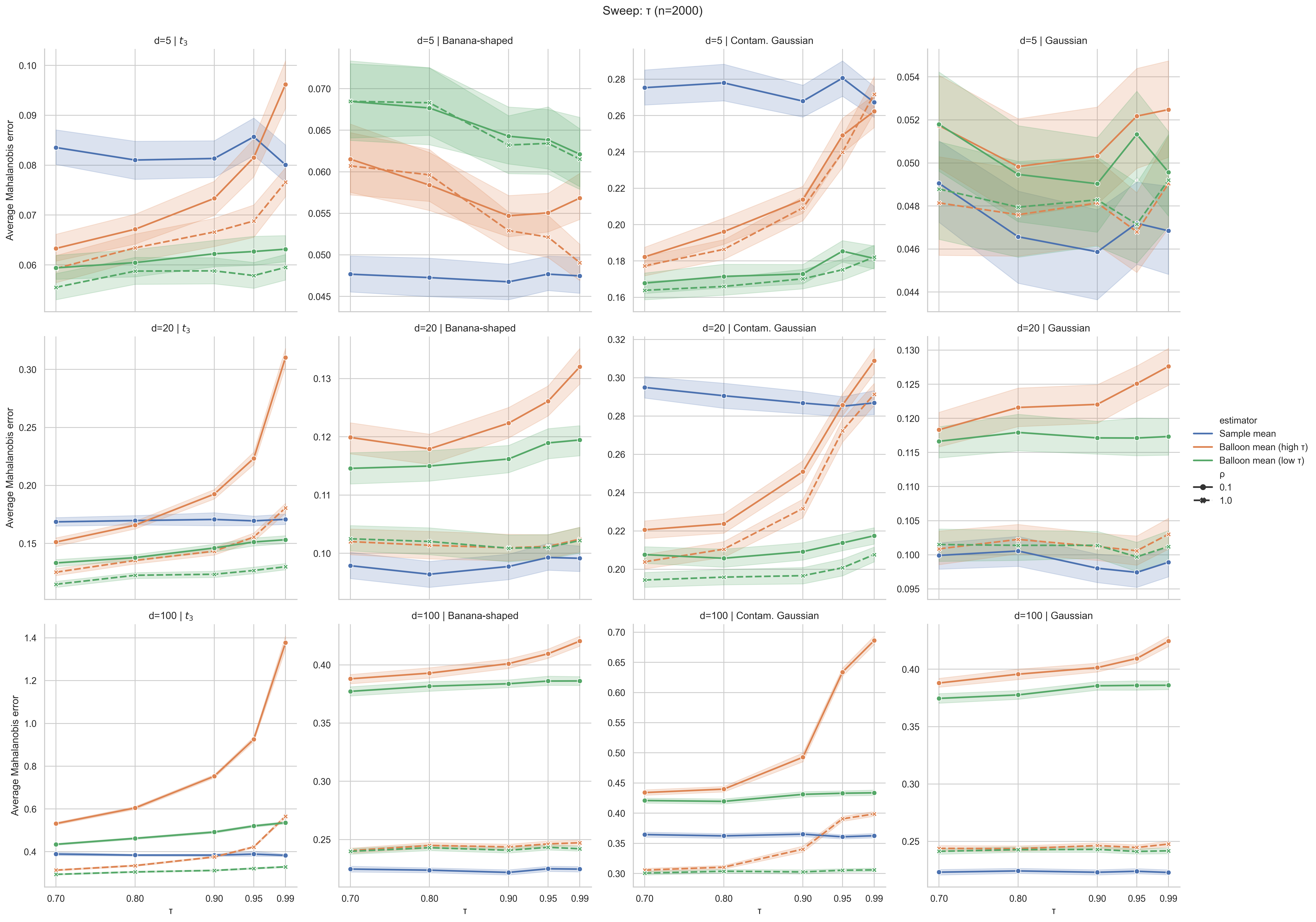}
        \caption{Average Mahalanobis error as a function of $\tau$ for mean estimation under four distributions and differing dimensions $d$ at $n=2000$. Solid lines show the mean error over repetitions, with shaded bands indicating a 95\% confidence interval over the simulation runs. We report the balloon mean with high-$\tau$ and low-$\tau$ settings, across privacy levels $\rho \in \{0.1,1.0\}$ (linetypes).}
    \label{fig:tau-2000}
\end{figure}
\FloatBarrier
\subsection{Effect of the number of iterations $M$}
Here we present the results for the effect of the number of iterations $M$. In low dimensions, the error appears to stabilize at three iterations ($M=3$), whereas in high dimensions, the error stabilizes at four iterations ($M=4$). This is consistent across generating distribution, sample size $n$, and does not depend on the $\tau$ variant. 
\begin{figure}[htbp]
\centering
    \includegraphics[width=\linewidth]{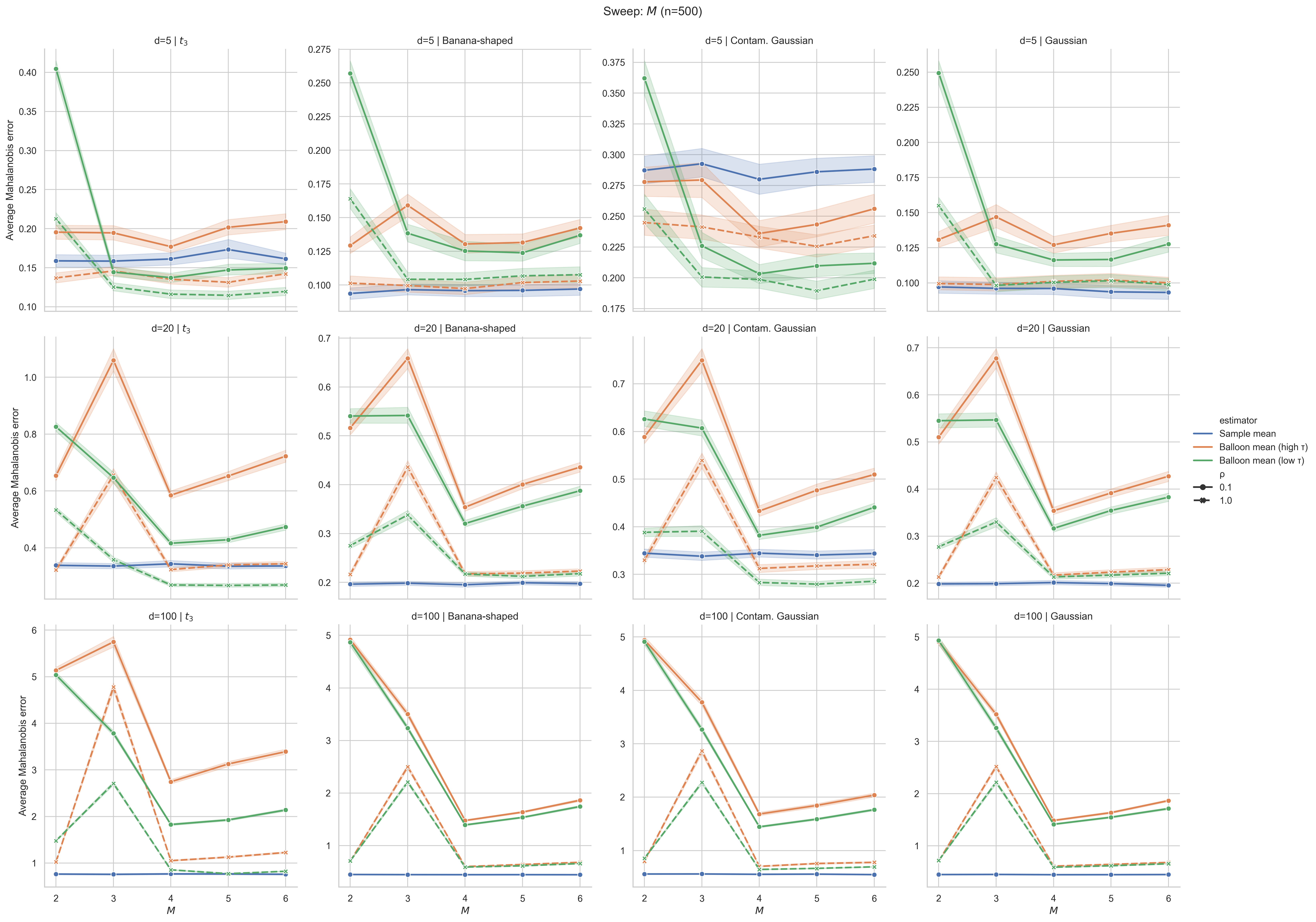}
            \caption{Average Mahalanobis error as a function of $M$ for mean estimation under four distributions and differing dimensions $d$ at $n=2000$. Solid lines show the mean error over repetitions, with shaded bands indicating a 95\% confidence interval over the simulation runs. We report the balloon mean with high-$\tau$ and low-$\tau$ settings, across privacy levels $\rho \in \{0.1,1.0\}$ (linetypes).}
    \label{fig:iterations-500}
\end{figure}
\begin{figure}[htbp]
\centering
    \includegraphics[width=\linewidth]{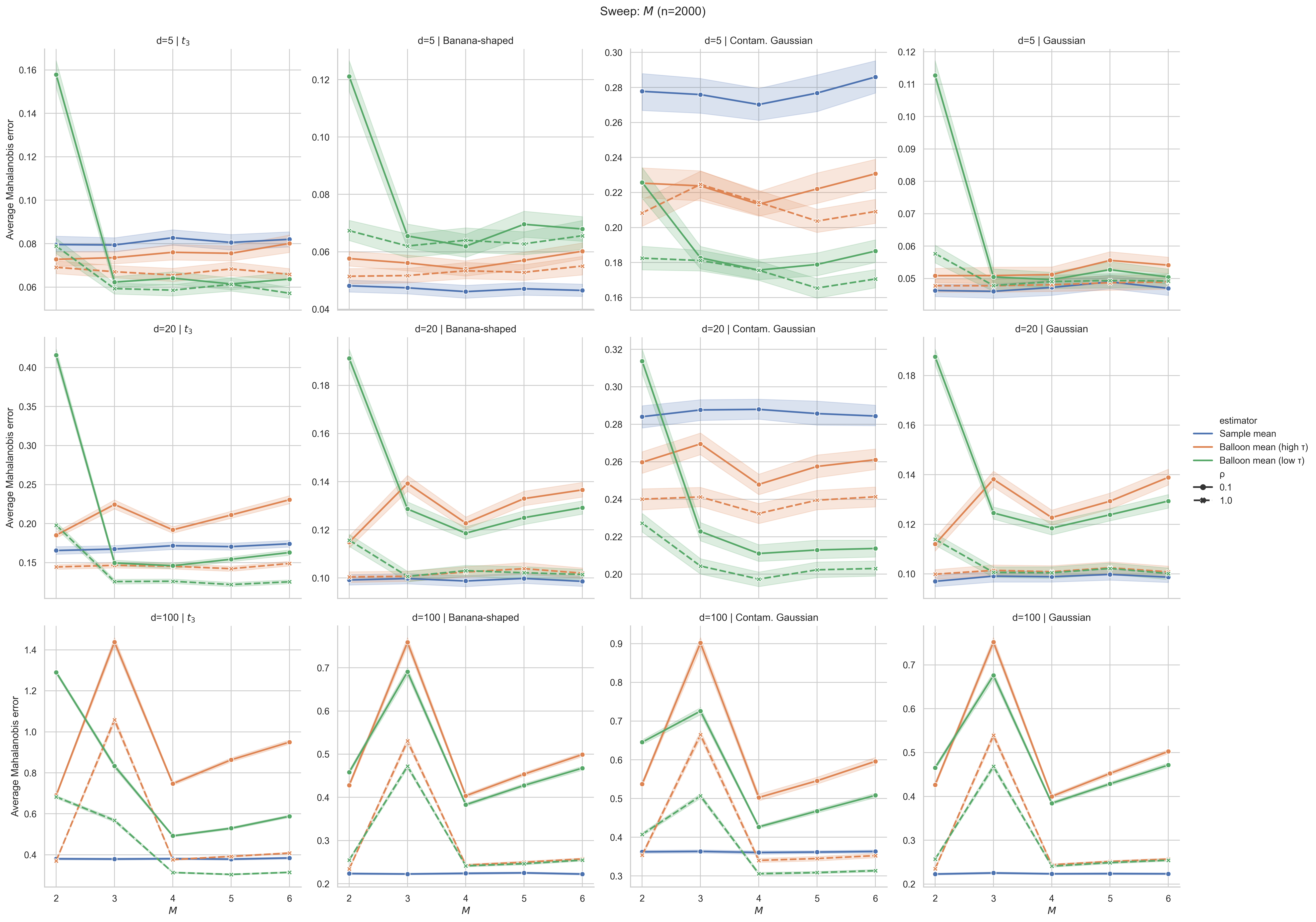}
           \caption{Average Mahalanobis error as a function of $M$ for mean estimation under four distributions and differing dimensions $d$ at $n=2000$. Solid lines show the mean error over repetitions, with shaded bands indicating a 95\% confidence interval over the simulation runs. We report the balloon mean with high-$\tau$ and low-$\tau$ settings, across privacy levels $\rho \in \{0.1,1.0\}$ (linetypes).}
    \label{fig:iterations-2000}
\end{figure}
\FloatBarrier
\section{Investigating different $\tau$ schedules}\label{app::tau-schedule}
To study the effect of the $\tau$-schedule across iterations, we fixed the number of iterations to $M=4$ and compared three types of schedules: constant, increasing, and decreasing. We centered these schedules around $\tau=0.9$. For the high-$\tau$ variant, we set $(\tau,\tau,\tau)$ for the constant schedule, $(\tau-0.2,\tau-0.1,\tau)$ for the increasing schedule, and $(\tau,\tau-0.1,\tau-0.2)$ for the decreasing schedule. For the low-$\tau$ variant, we did the same thing, instead with $\tau=0.6$. All experiments used $\beta=1.01$ and $\tilde R_0=50\sqrt d$. We varied the distribution (Gaussian, contaminated Gaussian, $t_3$, and banana), the dimension $d\in\{5,20,50,100\}$, the sample size $n\in\{500,2000\}$, and the privacy budget $\rho\in\{0.1,1\}$, and repeated each configuration for 50 trials. The results can be seen below in Figures~\ref{fig:tau-nonrobust-smallrho}--\ref{fig:tau-robust-largerho}. 

We first consider the high-$\tau$ (non-robust) variant. In contaminated and heavy-tailed settings, the decreasing-$\tau$ schedule consistently yields lower error, which is a result of the robustness awarded by smaller $\tau$ values. In the Gaussian and banana-shaped settings, decreasing $\tau$ is also preferred in higher-privacy regimes (smaller $\rho$) and in higher dimensions. Outside of these regimes, the three $\tau$ schedules perform similarly, and no strong separation is observed.

We next consider the low-$\tau$ (robust) variant. Here, a decreasing-$\tau$ schedule provides some additional robustness beyond the already conservative choice of $\tau$, but less so than in the high-$\tau$ variant. In the Gaussian and banana-shaped settings, however, the three $\tau$ schedules behave comparably, and we do not observe a consistent advantage for any single schedule.
\FloatBarrier
\begin{figure}[htbp]
    \centering
    \includegraphics[width=\textwidth]{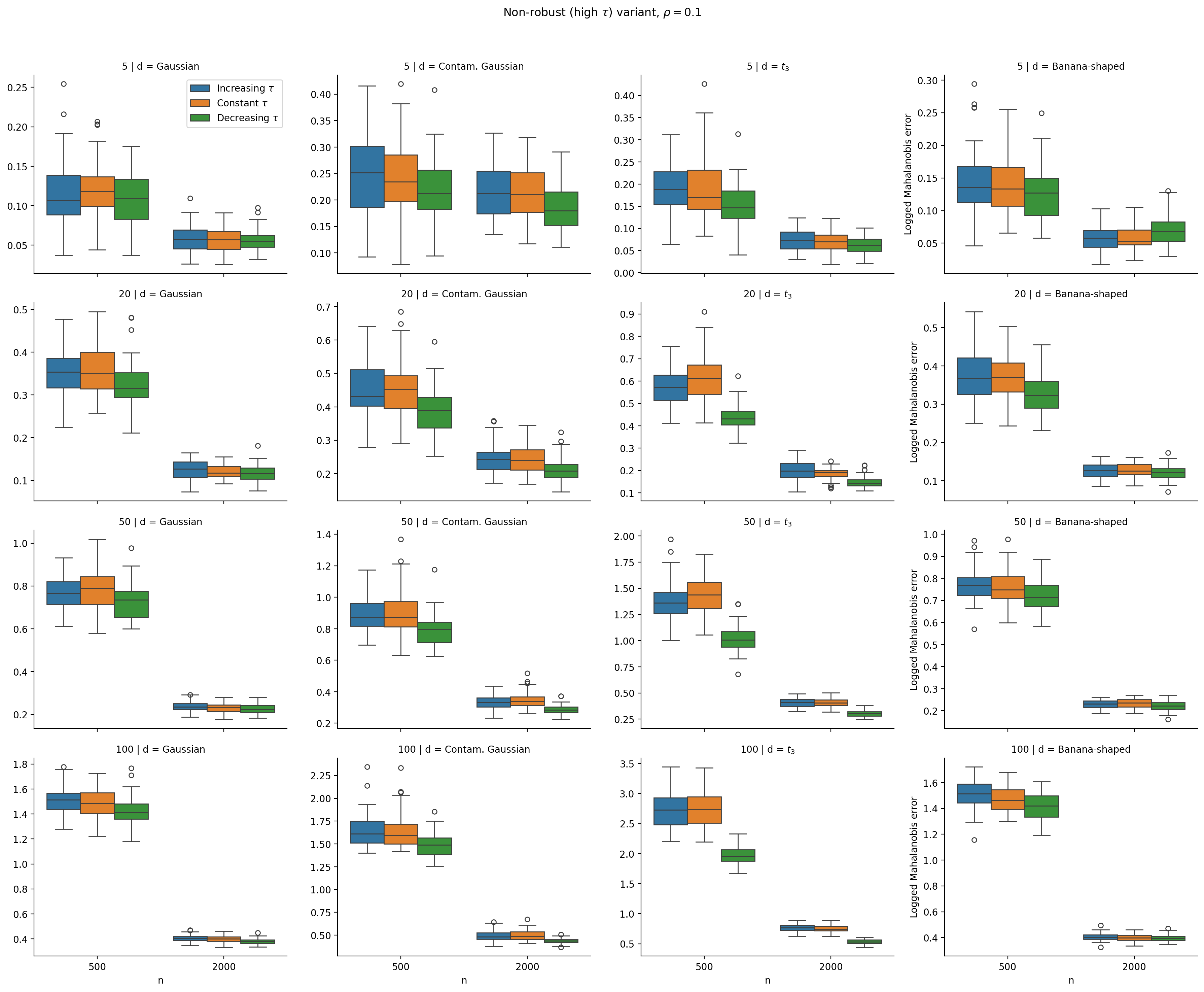}
    \caption{
    Non-robust balloon mean variants ($\rho = 0.1$). 
    Boxplots of the Mahalanobis error across distributions and dimensions.
    The three colors correspond to increasing, constant, and decreasing $\tau$ schedules.
    }
    \label{fig:tau-nonrobust-smallrho}
\end{figure}
\begin{figure}[htbp]
    \centering
    \includegraphics[width=\textwidth]{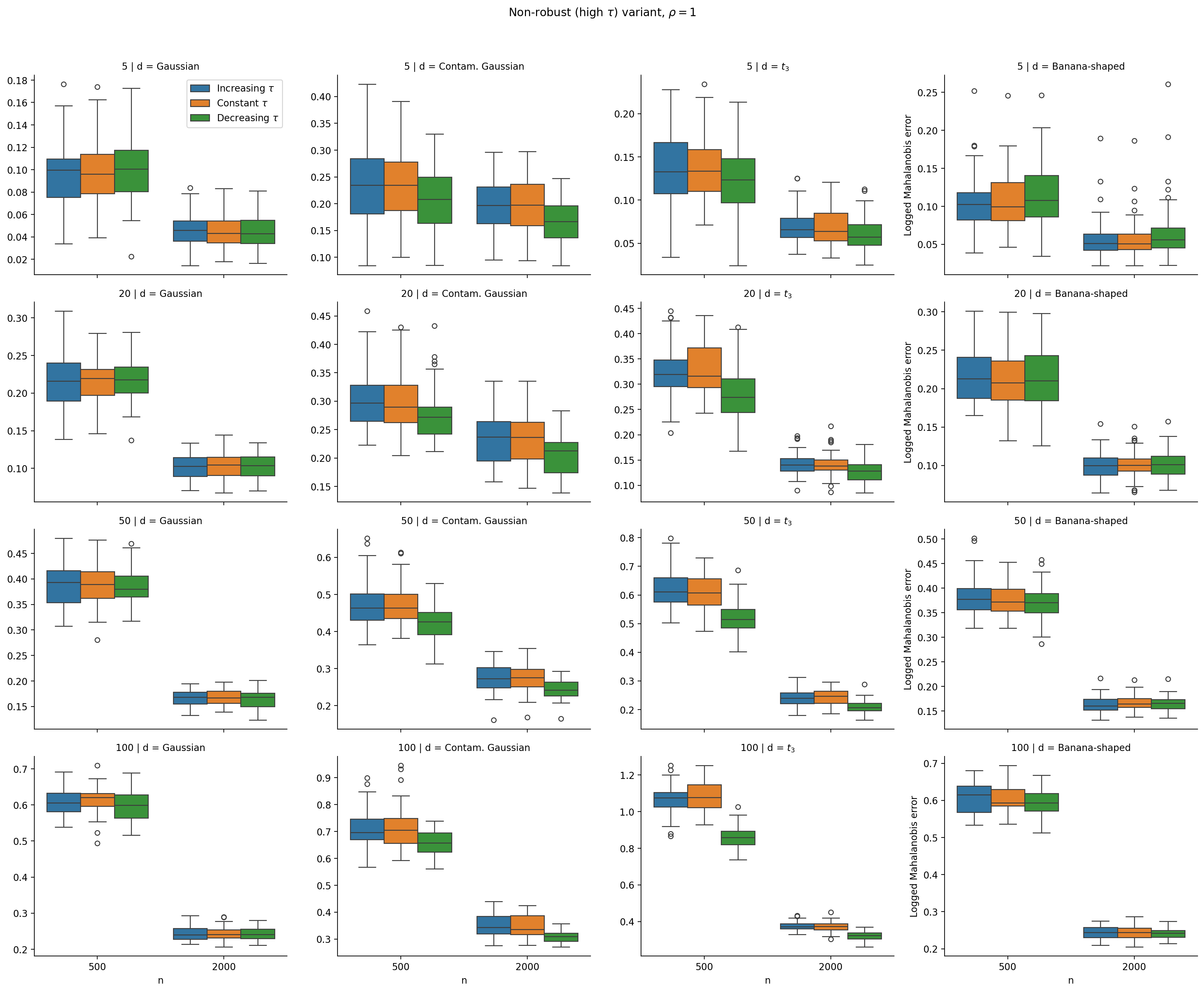}
    \caption{
    Non-robust balloon mean variants ($\rho = 1$). 
    Boxplots of the Mahalanobis error across distributions and dimensions.
    The three colors correspond to increasing, constant, and decreasing $\tau$ schedules.
    }
    \label{fig:tau-nonrobust-largerho}
\end{figure}
\begin{figure}[htbp]
    \centering
    \includegraphics[width=\textwidth]{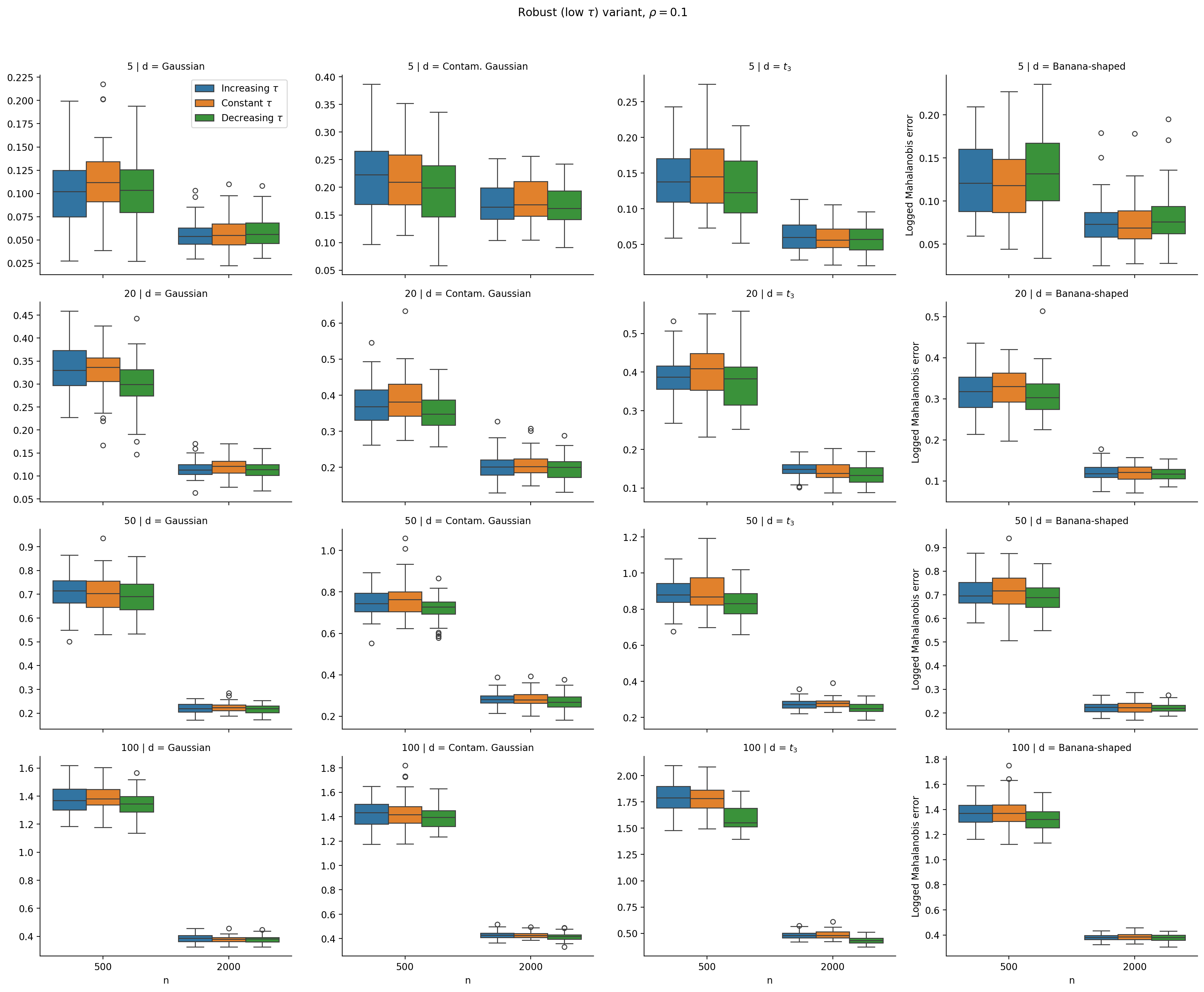}
    \caption{
    Robust balloon mean variants ($\rho = 0.1$). 
    Boxplots of the Mahalanobis error across distributions and dimensions.
    The three colors correspond to increasing, constant, and decreasing $\tau$ schedules.
    }
    \label{fig:tau-robust-smallrho}
\end{figure}
\begin{figure}[htbp]
    \centering
    \includegraphics[width=\textwidth]{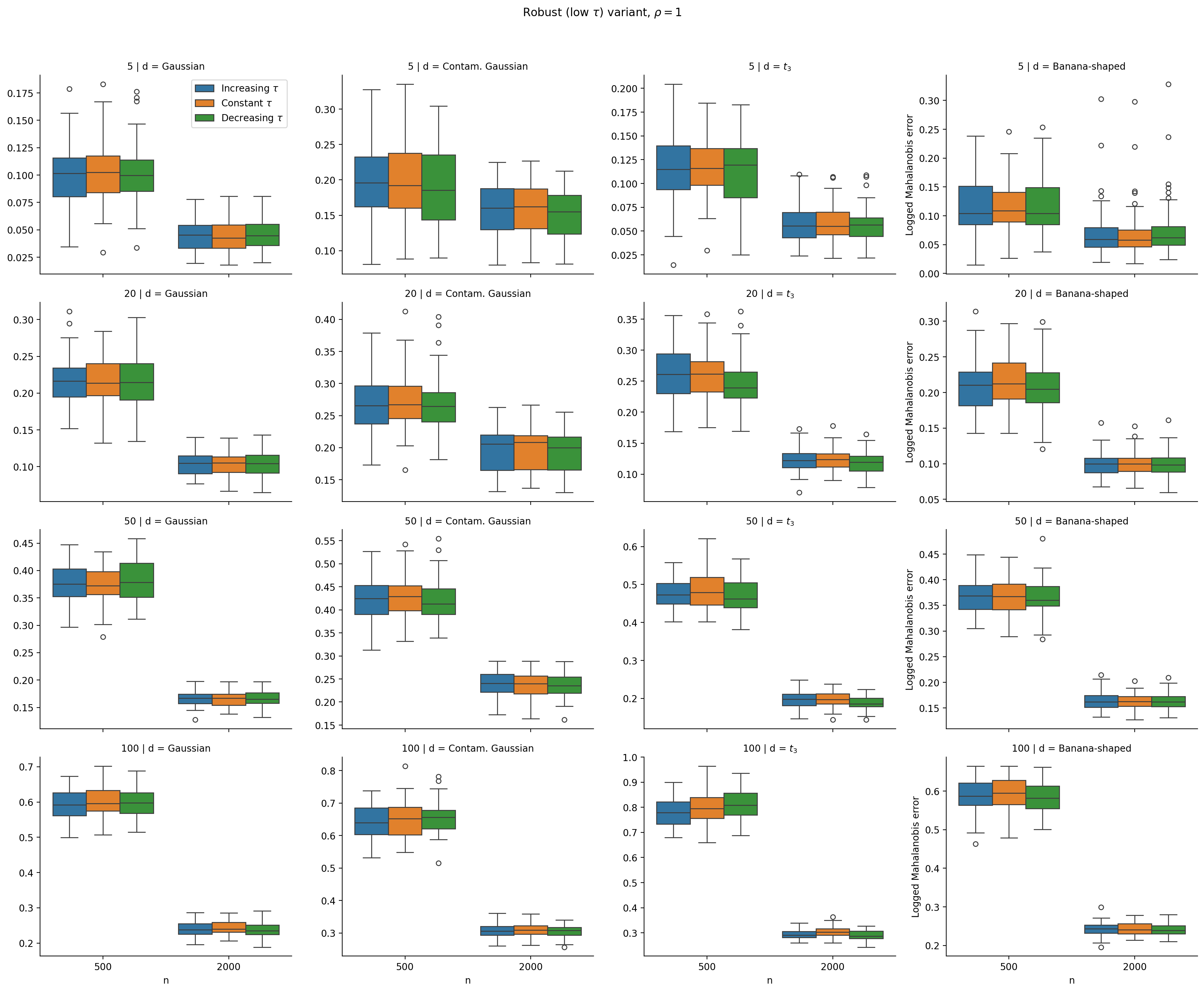}
    \caption{
    Robust balloon mean variants ($\rho = 1$). 
    Boxplots of the Mahalanobis error across distributions and dimensions.
    The three colors correspond to increasing, constant, and decreasing $\tau$ schedules.
    }
    \label{fig:tau-robust-largerho}
\end{figure}
\FloatBarrier



\end{document}